\documentclass[11pt]{article}

\usepackage[utf8]{inputenc}
\usepackage[english]{babel}
\usepackage[T1]{fontenc}
\usepackage{amsmath}
\usepackage{amssymb}
\usepackage{amsthm}
\usepackage{eucal}
\usepackage{qcircuit}
\usepackage{thmtools}
\usepackage[scaled=.92]{helvet}
\usepackage[
a4paper,
top=1in,
bottom=1in,
left=1in,
right=1in]{geometry}

\usepackage{mathtools}
\usepackage{microtype}
\usepackage{mathpazo}
\usepackage{xcolor}

\allowdisplaybreaks

\definecolor{citegreen}{HTML}{208054}
\definecolor{citeblue}{HTML}{0055cc}

\usepackage[hypertexnames=false]{hyperref}
\hypersetup{
    colorlinks=true,
    urlcolor=citeblue,
    linkcolor=citeblue,
    citecolor=citegreen,
}

\newtheorem{theorem}{Theorem}[section]
\newtheorem{proposition}[theorem]{Proposition}
\newtheorem{lemma}[theorem]{Lemma}
\newtheorem{corollary}[theorem]{Corollary}
\usepackage{authblk}

\renewcommand{\le}{\leqslant}
\renewcommand{\leq}{\leqslant}

\renewcommand{\geq}{\geqslant}

\title{Fast Random Compilers for Time-Dependent and Time-Independent Lindbladian Simulation}

\author[1]{Youngjin~Seo\thanks{\texttt{yjseo@kias.re.kr}}\protect\footnotemark[3]
}

\author[2,3]{Dhrumil~Patel\thanks{\texttt{djp265@vt.edu}}\thanks{The first two authors contributed equally.}
}
\author[1]{Hyukjoon~Kwon}
\author[4]{Mark~M.~Wilde}

\affil[1]{School of Computational Sciences, Korea Institute for Advanced Study, Seoul 02455, Republic of Korea}
\affil[2]{Department of Computer Science, Virginia Tech, Alexandria, Virginia 22305, USA}
\affil[3]{Department of Computer Science, Cornell University, Ithaca, New York 14850, USA}

\affil[4]{School of Electrical and Computer Engineering, Cornell University, Ithaca, New York 14850, USA}

\date{October 4, 2026}

\begin{document}

\maketitle

\begin{abstract}
    Randomized product formulas provide simple and flexible methods for quantum simulation. 
    While qDRIFT-type randomized simulation has been extended to Lindbladian dynamics, higher-order randomized compilers have primarily been developed for Hamiltonian simulation. 
    In this work, we develop first- and second-order randomized sampling algorithms for Lindbladian simulation in both the time-independent and time-dependent settings.
    Building on first-order qDRIFT-type Lindbladian simulation, we introduce a qSWIFT-type second-order correction for time-independent Lindbladians with a finite local decomposition. 
    For $\lambda t\geq 1$, where $\lambda$ is the sum of the Hamiltonian coefficients and dissipative rates, the resulting approximation achieves normalized diamond-norm error at most $\varepsilon$ using $r=O(\lambda^2t^2/\sqrt{\varepsilon})$ time slices, improving the $\varepsilon^{-1}$ precision dependence of the first-order qDRIFT-type approximation. 
    We then develop first- and second-order continuous-time sampling algorithms for general time-dependent Lindbladians, with the second-order construction likewise achieving an $\varepsilon^{-1/2}$ dependence on the target precision in the number of time slices. 
    We further extend the time-dependent construction to a linear-combination input model over local Hamiltonian and dissipative generators, allowing the randomized sampler to be formulated in terms of local evolutions rather than the full time-dependent Lindbladian.
    Because the second-order corrected maps need not be CPTP, we use them to estimate observable expectation values rather than to deterministically prepare output states.
\end{abstract}

\tableofcontents

\section{Introduction}

\subsection{Background and Motivation}

Quantum simulation is one of the central applications of quantum computation.
Starting from early product formula constructions for simulating local quantum dynamics~\cite{Lloyd1996UniversalQuantumSimulators}, Hamiltonian simulation has developed into a fundamental primitive in the theory of  quantum algorithms.
Major approaches include product-formula methods and refined Trotter error analyses~\cite{Suzuki1991GeneralTheoryFractal,Childs2021TheoryTrotterError}, sparse Hamiltonian simulation~\cite{Berry2007EfficientQuantumAlgorithms,Berry2015HamiltonianSimulationNearly}, quantum walk methods~\cite{Childs2010RelationshipBetweenContinuous,Berry2015HamiltonianSimulationNearly}, linear combination of unitaries and truncated Taylor series methods~\cite{Childs2012HamiltonianSimulationLinear,Berry2015SimulatingHamiltonianDynamics}, and quantum signal processing and qubitization~\cite{Low2017OptimalHamiltonianSimulation,Low2019HamiltonianSimulationQubitization}.
These developments have enabled increasingly efficient simulation of closed quantum systems across a broad range of algorithmic settings.

In many realistic physical settings, however, quantum systems are not closed. 
They interact with their surrounding environment, and their dynamics are more naturally described by quantum master equations~\cite{breuer2002TheoryOpenQuantum,rivas2012OpenQuantumSystems}. 
Under the Markovian approximation, the evolution is generated by a Lindbladian in Gorini--Kossakowski--Sudarshan--Lindblad form~\cite{Gorini1976CompletelyPositive,Lindblad1976GeneratorsQuantumDynamical}.
Such dynamics arise in quantum optics and many-body systems, and they also provide a framework for describing noise and engineered dissipation in quantum devices~\cite{wiseman2009QuantumMeasurementControl,weiss2012QuantumDissipativeSystems,poyatos1996QuantumReservoirEngineering,verstraete2009QuantumComputationDissipation}. 
Efficiently simulating Lindbladian dynamics is therefore important both for modeling open quantum systems and for using dissipation as an algorithmic resource.

Quantum algorithms for Lindbladian simulation have been developed in several directions. 
Early work studied efficient product-formula simulation of local Lindbladian dynamics~\cite{Kliesch2011DissipativeQuantumChurch}. 
Subsequent work developed efficient algorithms for sparse Markovian quantum dynamics~\cite{Childs2017EfficientSimulationSparse}, while algorithms based on linear-combination techniques for quantum channels achieved improved scaling in the simulation time and precision~\cite{Cleve2017EfficientQuantumAlgorithms}.
Wave matrix Lindbladization (WML) introduced a sample-based input model in which Lindblad operators are encoded into quantum program states, with extensions to general Lindbladians and linear-combinations and polynomials~\cite{Patel2023WMLI,Patel2023WMLII}. 
The sampling-based WML algorithm was subsequently implemented for digital simulation of the nonresonant open Tavis--Cummings model~\cite{Sims2025DigitalQuantumSimulations}.
Higher-order series-expansion methods were later developed for Markovian open-system simulation~\cite{Li2023SimulatingMarkovianOpen}.
More recent works have explored randomized, trajectory-inspired, and product-formula-based approaches to Lindblad evolution~\cite{Chen2025RandomizedMethodSimulating,Peng2025QuantumTrajectoryInspired,Borras2025QuantumAlgorithmSimulate,David2026FasterQuantumSimulation}.

In parallel, randomized methods have become an important approach to Hamiltonian simulation~\cite{Childs2019FasterQuantumSimulation,Campbell2019RandomCompiler}. 
The qDRIFT algorithm shows that a random product formula can simulate Hamiltonian dynamics with a cost independent of the number of Hamiltonian terms, depending instead on the one norm of the coefficients~\cite{Campbell2019RandomCompiler}. 
Higher-order randomized methods, such as qSWIFT and randomized multi-product formulas, further improve the dependence on the target precision while retaining a simple sampling-based structure~\cite{Nakaji2024HighOrderRandomized,Faehrmann2022RandomizingMultiProduct}. 
These developments motivate the question of whether analogous randomized compilers can be developed for open-system dynamics.

Time dependence introduces another essential challenge. 
External driving and control naturally lead to time-dependent Hamiltonians in closed systems. 
In open systems, time-dependent Hamiltonian terms, dissipation rates, or jump operators lead to time-dependent Lindbladians~\cite{barthel2012QuasilocalityEfficient}. 
For closed systems, time-dependent Hamiltonian simulation has been studied using truncated Dyson series~\cite{Kieferova2019SimulatingDynamicsTimeDependent}, interaction-picture methods~\cite{Low2018HamiltonianSimulationInteraction}, permutation expansions~\cite{Chen2021QuantumAlgorithmTimeDependent}, and algorithms with $L^1$-norm scaling~\cite{Berry2020TimeDependentHamiltonian}. 
For open systems, time-dependent Lindbladian simulation is more delicate because the evolution is described by a time-ordered exponential of nonunitary generators and the approximation must preserve, or at least carefully control, completely positive trace-preserving dynamics~\cite{barthel2012QuasilocalityEfficient,Li2023SimulatingMarkovianOpen,Borras2025QuantumAlgorithmSimulate}.

\subsection{Summary of Results}

In this work, we develop randomized sampling algorithms for both time-independent and time-dependent Lindbladian simulation. 
First, for time-independent Lindbladians, we recall a qDRIFT-type averaged channel, following the extension of qDRIFT ideas from Hamiltonian simulation to Lindbladian dynamics~\cite{Campbell2019RandomCompiler,Chen2025RandomizedMethodSimulating}. 
We then introduce a qSWIFT-type second-order correction for Lindbladian dynamics, inspired by higher-order randomized compilers for Hamiltonian simulation~\cite{Nakaji2024HighOrderRandomized}. 
This correction is obtained by comparing the Taylor expansion of the exact Lindblad evolution with that of the averaged qDRIFT channel, yielding an averaged approximation whose normalized diamond-norm error scales as $O(1/r^2)$ in the number of time slices.
We also briefly discuss how the same expansion suggests a hierarchy of higher-order qSWIFT-type corrections.

Second, we consider general time-dependent Lindbladians. 
We formulate a first-order continuous-time sampling algorithm whose averaged channel matches the first-order term of the time-ordered evolution. 
We then construct a second-order correction by inserting a correction map into one time slice, obtaining improved $O(1/r^2)$ scaling in normalized diamond norm. 
We further extend the construction to a linear-combination input model, in which the time-dependent Lindbladian is decomposed into local Hamiltonian and dissipative generators and the algorithm samples both a local term and a time within each interval.

Finally, because the second-order corrected maps are not necessarily CPTP channels, we describe how to estimate observable expectation values using swift-map-based Hadamard-test-type constructions. 
Under the stated implementation assumptions, these procedures give unbiased estimators for the corresponding second-order approximations.
Our time-slice bounds assume access to the required sampling procedures, evolution and dilation primitives, state preparation, and observable measurements. 
Classical preprocessing and finite-precision synthesis costs are not included in the time-slice count.

\section{Time-Independent Lindbladian Simulation}

\subsection{Problem Statement}

We begin by considering the simulation of a time-independent Lindblad master equation
\begin{equation}
    \frac{d \rho}{dt} = \mathcal{L}(\rho) \coloneqq -i[H, \rho] + \sum_{k=1}^{K} \ell_k \left(L_k\rho L_k^\dagger - \frac{1}{2} \{L_k^\dagger L_k, \rho\}\right) \label{eq:time-indep-Lind},
\end{equation}
where $H$ is a time-independent Hamiltonian, each $L_k$ is a normalized Lindblad operator satisfying $\Vert L_k\Vert=1$, and $\ell_k> 0$ is the corresponding dissipative rate. 
Equivalently, the $k$th physical jump operator can be written as $\sqrt{\ell_k}L_k$.
The superoperator $\mathcal{L}$ in \eqref{eq:time-indep-Lind} is called a Lindbladian, and the notation $\{A, B\}$ refers to the anti-commutator of operators $A$ and $B$, i.e., $\{A, B\} \coloneqq AB + BA$. 

We assume that the Hamiltonian admits a local decomposition
\begin{equation}
    H \coloneqq \sum_{j=1}^{J} h_{j} H_{j},\label{eq:H-def}
\end{equation}
where $h_{j}>0$, $ \Vert H_{j} \Vert = 1$, and each $H_{j}$ acts nontrivially on at most a constant number of qubits.  
We also assume that each Lindblad operator $L_k$ is local, in the sense that it acts nontrivially on at most a constant number of qubits.
Such locality assumptions are standard in the simulation of open many-body systems governed by local or short-range Lindbladians~\cite{barthel2012QuasilocalityEfficient}.

By simulating the Lindbladian evolution for time $t$, we mean implementing, or approximating, the quantum channel $e^{t\mathcal{L}}$, which is the solution map of \eqref{eq:time-indep-Lind}, where
\begin{equation}
    e^{t\mathcal{L}}(\rho) = \sum_{s=0}^\infty \frac{t^s}{s!}\mathcal{L}^s(\rho) ,
    \label{eq:Lind-expand}
\end{equation}
and $\mathcal{L}^s$ denotes $s$ sequential applications of the superoperator $\mathcal{L}$. 
The equality above simply comes from the Taylor series expansion of the exponential. 
We measure the cost of an algorithm by counting the number of local operations $\{e^{-i\tau H_{j}}\}_j$ and $\{e^{\tau \mathcal{L}_{k}}\}_k$, where $\{\mathcal{L}_{k}\}_k$ are the local single-operator dissipative Lindbladians defined as
\begin{align}
    \mathcal{L}_{k}(\rho) \coloneqq L_{k}\rho L_{k}^{\dagger} - \frac{1}{2} \left \{L_{k}^{\dagger}L_{k}, \rho \right\}.
\end{align}

\subsection{First-Order qDRIFT-Type Sampling Algorithm}

In this subsection, we recall a qDRIFT-type sampling algorithm for Lindbladian simulation~\cite{Chen2025RandomizedMethodSimulating} and specialize it to the finite local decomposition considered in this work.

\medskip
\noindent
\textbf{qDRIFT-type Sampling Algorithm --- } Let $\lambda \coloneqq \sum_{j=1}^{J} h_{j} + \sum_{k=1}^{K} \ell_{k}$ and $\tau \coloneqq t/r$. Starting from an input state $\rho$, repeat the following steps $r$ times:
\begin{enumerate}
    \item Sample an index $i\in \{1,\ldots,J+K\}$ from the following probability distribution:
    \begin{align}
        p_{i} = 
        \begin{cases}
          \frac{h_{i}}{\lambda}, & i \in \{ 1, \ldots, J \},\\
          \frac{\ell_{i-J}}{\lambda}, & i \in \{ J+1, \ldots, J+K\}.
        \end{cases}
        \label{eq:prob-den-time-ind}
    \end{align}
    \item If $i \in \{1, \ldots, J\}$, then apply the unitary channel $\rho \mapsto e^{-iH_i\lambda\tau}\rho e^{iH_i\lambda\tau}$.
    \item Else if $i \in \{J+1, \ldots, J+K\}$, then apply the quantum channel $e^{\lambda\tau\mathcal{L}_{i-J}}$ to the input state.
\end{enumerate}
Return the final state.

The averaged one-step channel associated with this procedure is
\begin{equation}
\mathcal{Q}_{\tau}(\rho)
\coloneqq
\sum_{j=1}^{J}
\frac{h_j}{\lambda}
e^{-iH_j\lambda\tau}\rho e^{iH_j\lambda\tau} +
\sum_{k=1}^{K}
\frac{\ell_k}{\lambda}
e^{\lambda\tau\mathcal{L}_k}(\rho).
\label{eq:qdrift-one-step-channel}
\end{equation}
Thus, the averaged output channel after $r$ steps is $\mathcal{Q}_{\tau}^{r}$.

The following bound follows by specializing the average-channel analysis of~\cite{Chen2025RandomizedMethodSimulating} to the finite local decomposition considered here.
\begin{proposition}[qDRIFT-type Lindbladian simulation]
\label{prop:qdrift-t-algo-gc}
Let $\mathcal{L}$ be the Lindbladian in~\eqref{eq:time-indep-Lind}, and let $\mathcal{Q}_{\tau}$ be the averaged one-step channel in~\eqref{eq:qdrift-one-step-channel}.
If $\tau=t/r$ and
\begin{equation}
r \geq \frac{2\lambda^2t^2}{\varepsilon},
\end{equation}
then
\begin{equation}
\frac{1}{2}
\left\Vert
e^{t\mathcal{L}} -\mathcal{Q}_{\tau}^{r}
\right\Vert_{\diamond}
\leq
\varepsilon.
\end{equation}
\end{proposition}
\begin{proof}
   See Appendix~\ref{app:qdrift-t-algo-gc}.
\end{proof}

\subsection{Second-Order qSWIFT-Type Sampling Algorithm}

In the previous subsection, we recalled a qDRIFT-type first-order approximation for simulating the Lindbladian evolution $e^{t\mathcal{L}}$~\cite{Chen2025RandomizedMethodSimulating} and specialized it to the finite local decomposition considered in this work. 
We now introduce a qSWIFT-type second-order correction to this averaged qDRIFT construction. 
In particular, we focus on a second-order correction, which improves the error dependence from $O(1/r)$ to $O(1/r^2)$ at the level of the averaged approximation map.

The idea is adapted from the qSWIFT algorithm originally developed for Hamiltonian simulation~\cite{Nakaji2024HighOrderRandomized}. 
Here, we formulate an analogous correction for Lindbladian dynamics by comparing the Taylor expansion of the exact one-step evolution $e^{\tau\mathcal{L}}$ with that of the averaged qDRIFT one-step channel. 
This correction will be used to construct a second-order approximation to the target Lindbladian evolution.

Recall that $\mathcal{Q}_{\tau}$ denotes the averaged one-step qDRIFT channel. For brevity, we introduce elementary generators $\mathcal{M}_{\ell}$ so that the Lindbladian can be written as
\begin{equation}
    \mathcal{L} = \lambda \sum_{\ell=1}^{J+K} p_{\ell}\mathcal{M}_{\ell},
\end{equation}
where $\mathcal{M}_{\ell}$ is either the Hamiltonian generator $\mathcal{H}_j(\rho)\coloneqq -i[H_j,\rho]$ or the single-operator dissipative Lindbladian $\mathcal{L}_k$, depending on the index $\ell$, and $p_{\ell}$ is the sampling probability defined in \eqref{eq:prob-den-time-ind}.
Using this, we can rewrite the qDRIFT channel as
\begin{equation}
    \mathcal{Q}_{\tau} =  \sum_{\ell} p_{\ell} e^{\lambda\tau \mathcal{M}_{\ell}},
\end{equation}
where $\tau \coloneqq t/r$. Using the Taylor expansion of the exponential function, we get
\begin{align}
    \mathcal{Q}_{\tau} & =  \mathcal{I} + \sum_{\ell} p_{\ell} \lambda\tau \mathcal{M}_{\ell} + \sum_{\ell} p_{\ell} \sum_{m=2}^{\infty} \frac{(\lambda\tau )^m}{m!} \mathcal{M}_{\ell}^m\\
    & = \mathcal{I} + \tau \mathcal{L} + \sum_{\ell} p_{\ell} \sum_{m=2}^{\infty} \frac{(\lambda\tau )^m}{m!} \mathcal{M}_{\ell}^m.
\end{align}
Similarly, we can expand the target map $e^{\tau \mathcal{L}}$ as
\begin{equation}
    e^{\tau \mathcal{L}} = \mathcal{I} + \tau \mathcal{L} + \sum_{m=2}^{\infty} \frac{\tau^m}{m!} \mathcal{L}^m.
\end{equation}
Subtracting the expansion of $\mathcal{Q}_{\tau}$ from that of $e^{\tau\mathcal{L}}$, we obtain
\begin{align}
    e^{\tau \mathcal{L}}  & =  \mathcal{Q}_{\tau} + \sum_{m=2}^{\infty} \frac{\tau^m}{m!} \mathcal{L}^m -  \sum_{\ell} p_{\ell} \sum_{m=2}^{\infty} \frac{\lambda^m\tau^m}{m!} \mathcal{M}_{\ell}^m \\
    & = \mathcal{Q}_{\tau} + \sum_{m=2}^{\infty} \frac{\lambda^m\tau^m}{m!} \left  (\left(\frac{\mathcal{L}}{\lambda}\right)^m - \sum_{\ell} p_{\ell} \mathcal{M}_{\ell}^m\right)\\
    & = \mathcal{Q}_{\tau} + \sum_{m=2}^{\infty} \frac{\lambda^m\tau^m}{m!} \left  (\widetilde{\mathcal{L}}^m - \sum_{\ell} p_{\ell} \mathcal{M}_{\ell}^m\right)\\
    & = \mathcal{Q}_{\tau} + \sum_{m=2}^{\infty} \frac{\lambda^m\tau^m}{m!}  \mathcal{L}^{(m)},
    \label{eq:qswift-correction}
\end{align}
where we define $\widetilde{\mathcal{L}} \coloneqq \mathcal{L}/\lambda$ and
\begin{equation}
      \mathcal{L}^{(m)} \coloneqq \widetilde{\mathcal{L}}^m - \sum_{\ell} p_{\ell} \mathcal{M}_{\ell}^m.\label{eq:qswift-generator-approx}
\end{equation}
From \eqref{eq:qswift-correction}, the leading correction to the averaged qDRIFT one-step map $\mathcal{Q}_{\tau}$ is the second-order term
\begin{equation}
    \frac{\lambda^2\tau^2}{2}\mathcal{L}^{(2)}.
\end{equation}
Thus, a second-order correction to the $r$-step averaged qDRIFT map is obtained by inserting this term at each of the $r$ possible positions, while leaving the remaining $r-1$ time slices as $\mathcal{Q}_{\tau}$, where products of linear maps are understood as compositions with the rightmost map acting first.
We therefore define the second-order qSWIFT map as
\begin{equation}
    \mathcal{Q}^{(2)}_t
    \coloneqq
    \mathcal{Q}_{\tau}^{r} + \frac{\lambda^2\tau^2}{2}\sum_{k=0}^{r-1}\mathcal{Q}_{\tau}^{r-1-k}\mathcal{L}^{(2)}\mathcal{Q}_{\tau}^{k}.
    \label{def:qswift-map}
\end{equation}
The summation accounts for all possible positions of a single second-order correction among the $r$ qDRIFT steps.

Next, in the following theorem, we analyze how well the second-order qSWIFT-type map approximates the target Lindbladian evolution $e^{t\mathcal{L}}$.

\begin{theorem}[Second-order qSWIFT-type Lindbladian approximation]
\label{thm:qswift-map-error}
Let $\mathcal{L}$ be the time-independent Lindbladian in \eqref{eq:time-indep-Lind}, and let $\mathcal{Q}^{(2)}_t$ be the second-order qSWIFT-type map defined in \eqref{def:qswift-map}. Then, for $\tau=t/r$, the normalized diamond-norm error satisfies
\begin{equation}
    \frac{1}{2}\left\Vert e^{t\mathcal{L}}-\mathcal{Q}^{(2)}_t\right\Vert_{\diamond}
    \leq
    \frac{4\lambda^3t^3}{3r^2} + \frac{4(r-1)\lambda^4t^4}{r^3}.
\end{equation}
Consequently,
\begin{equation}
    \frac{1}{2}\left\Vert e^{t\mathcal{L}}-\mathcal{Q}^{(2)}_t\right\Vert_{\diamond}
    = O\!\left(\frac{\lambda^3t^3+\lambda^4t^4}{r^2}\right).
\end{equation}
In particular, when $\lambda t\geq 1$, it suffices to choose
\begin{equation}
    r = O\!\left(\frac{\lambda^2t^2}{\sqrt{\varepsilon}}\right)
\end{equation}
to ensure
\begin{equation}
    \frac{1}{2}\left\Vert e^{t\mathcal{L}}-\mathcal{Q}^{(2)}_t\right\Vert_{\diamond}
    \leq\varepsilon.
\end{equation}
\end{theorem}

\begin{proof}
We first bound the one-step error that remains after including the second-order correction. 
From the definition of $\mathcal{L}^{(2)}$ in \eqref{eq:qswift-generator-approx},
\begin{equation}
    \frac{\lambda^2\tau^2}{2}\mathcal{L}^{(2)} =
    \frac{\tau^2}{2}\mathcal{L}^2 - \frac{\lambda^2\tau^2}{2}\sum_{\ell}p_{\ell}\mathcal{M}_{\ell}^2.
    \label{eq:qswift-second-order-correction}
\end{equation}
Thus, this correction exactly accounts for the difference between the second-order terms of the exact one-step evolution and the averaged qDRIFT channel. 
Using $\mathcal{L}=\lambda\sum_{\ell}p_{\ell}\mathcal{M}_{\ell}$ and $\sum_{\ell}p_{\ell}=1$, we can write
\begin{equation}
    e^{\tau\mathcal{L}}-\mathcal{Q}_{\tau}-\frac{\lambda^2\tau^2}{2}\mathcal{L}^{(2)}=
    \left(e^{\tau\mathcal{L}}-\mathcal{I}-\tau\mathcal{L}-\frac{\tau^2}{2}\mathcal{L}^2\right) -\sum_{\ell}p_{\ell}\left(e^{\lambda\tau\mathcal{M}_{\ell}}-\mathcal{I}-\lambda\tau\mathcal{M}_{\ell}-\frac{\lambda^2\tau^2}{2}\mathcal{M}_{\ell}^2\right).
    \label{eq:qswift-one-step-residual-expansion}
\end{equation}

For any Lindbladian generator $\mathcal{G}$ and $u\geq0$, Taylor's formula with integral remainder gives
\begin{equation}
    e^{u\mathcal{G}}-\mathcal{I}-u\mathcal{G}-\frac{u^2}{2}\mathcal{G}^2 
    = \int_0^u dv\,\frac{(u-v)^2}{2}e^{v\mathcal{G}}\circ\mathcal{G}^3.
\end{equation}
Since $e^{v\mathcal{G}}$ is a quantum channel for $v\geq0$, the same argument as in the first-order analysis gives
\begin{equation}
    \left\Vert e^{u\mathcal{G}}-\mathcal{I}-u\mathcal{G}-\frac{u^2}{2}\mathcal{G}^2 \right\Vert_{\diamond}
    \leq
    \frac{u^3}{6}\left\Vert\mathcal{G}\right\Vert_{\diamond}^3.
\end{equation}

Applying this bound to the two types of remainder terms in \eqref{eq:qswift-one-step-residual-expansion}, and using $\left\Vert\mathcal{L}\right\Vert_{\diamond}\leq2\lambda$, $\left\Vert\mathcal{M}_{\ell}\right\Vert_{\diamond}\leq2$, and $\sum_{\ell}p_{\ell}=1$, we obtain
\begin{align}
    \left\Vert e^{\tau\mathcal{L}}-\mathcal{Q}_{\tau}-\frac{\lambda^2\tau^2}{2}\mathcal{L}^{(2)} \right\Vert_{\diamond}
    &\leq
    \frac{\tau^3}{6}\left\Vert\mathcal{L}\right\Vert_{\diamond}^3 + \frac{\lambda^3\tau^3}{6}\sum_{\ell}p_{\ell}\left\Vert\mathcal{M}_{\ell}\right\Vert_{\diamond}^3\\
    &\leq
    \frac{8}{3}\lambda^3\tau^3.
    \label{eq:qswift-one-step-residual}
\end{align}

Using~\eqref{eq:qswift-second-order-correction}, together with the triangle inequality, we have
\begin{align}
    \left\Vert \frac{\lambda^2\tau^2}{2}\mathcal{L}^{(2)} \right\Vert_{\diamond}
    &\leq
    \frac{\lambda^2\tau^2}{2}\left(\left\Vert\frac{\mathcal{L}}{\lambda}\right\Vert_{\diamond}^2
    + \sum_{\ell}p_{\ell}\left\Vert\mathcal{M}_{\ell}\right\Vert_{\diamond}^2\right)\\
    &\leq 4\lambda^2\tau^2.
    \label{eq:qswift-second-order-correction-bound}
\end{align}

We now compare the full $r$-step exact evolution with the second-order qSWIFT-type map. 
From \eqref{def:qswift-map},
\begin{equation}
    e^{t\mathcal{L}}-\mathcal{Q}^{(2)}_t
    = \left[\left(e^{\tau\mathcal{L}}\right)^r-\mathcal{Q}_{\tau}^{r}\right] - \sum_{i=0}^{r-1}\mathcal{Q}_{\tau}^{r-1-i} \frac{\lambda^2\tau^2}{2}\mathcal{L}^{(2)}\mathcal{Q}_{\tau}^{i}.
\end{equation}
Applying the telescoping expansion to the bracketed difference and then adding and subtracting the same single-insertion term with $\mathcal{Q}_{\tau}^{r-1-i}$ replaced by $\left(e^{\tau\mathcal{L}}\right)^{r-1-i}$, we obtain
\begin{align}
    e^{t\mathcal{L}}-\mathcal{Q}^{(2)}_t
    &= \sum_{i=0}^{r-1}\left(e^{\tau\mathcal{L}}\right)^{r-1-i} \left(e^{\tau\mathcal{L}}-\mathcal{Q}_{\tau}-\frac{\lambda^2\tau^2}{2}\mathcal{L}^{(2)}\right) \mathcal{Q}_{\tau}^{i} \notag\\
    &\quad
    + \sum_{i=0}^{r-1}\left[\left(e^{\tau\mathcal{L}}\right)^{r-1-i}-\mathcal{Q}_{\tau}^{r-1-i}\right]
    \frac{\lambda^2\tau^2}{2}\mathcal{L}^{(2)}\mathcal{Q}_{\tau}^{i}.
    \label{eq:qswift-global-error-decomposition}
\end{align}
The first sum contains the residual one-step error after subtracting the second-order correction. The second sum accounts for the difference between the exact and qDRIFT evolutions following the inserted correction.
Since both $e^{\tau\mathcal{L}}$ and $\mathcal{Q}_{\tau}$ are quantum channels, the first sum is bounded using \eqref{eq:qswift-one-step-residual} as
\begin{equation}
    \frac{1}{2} \left\Vert\sum_{i=0}^{r-1}\left(e^{\tau\mathcal{L}}\right)^{r-1-i}\left(e^{\tau\mathcal{L}} - \mathcal{Q}_{\tau} - \frac{\lambda^2\tau^2}{2}\mathcal{L}^{(2)}\right) \mathcal{Q}_{\tau}^{i}\right\Vert_{\diamond}
    \leq \frac{4}{3}r\lambda^3\tau^3.
\end{equation}
For the second sum, we use the one-step qDRIFT bound
\begin{equation}
    \left\Vert e^{\tau\mathcal{L}}-\mathcal{Q}_{\tau}\right\Vert_{\diamond}
    \leq 4\lambda^2\tau^2,
\end{equation}
which follows from the first-order Taylor-remainder analysis used in the proof of Proposition~\ref{prop:qdrift-t-algo-gc}, since the exact and averaged qDRIFT one-step maps agree through first order.
Therefore, another telescoping expansion gives, for every nonnegative integer $m$,
\begin{equation}
    \left\Vert\left(e^{\tau\mathcal{L}}\right)^m -\mathcal{Q}_{\tau}^{m}\right\Vert_{\diamond}
    \leq m\left\Vert e^{\tau\mathcal{L}}-\mathcal{Q}_{\tau}\right\Vert_{\diamond}
    \leq 4m\lambda^2\tau^2.
\end{equation}
Together with \eqref{eq:qswift-second-order-correction-bound}, this gives
\begin{align}
    \frac{1}{2}\left\Vert e^{t\mathcal{L}}-\mathcal{Q}^{(2)}_t\right\Vert_{\diamond}
    &\leq
    \frac{4}{3}r\lambda^3\tau^3 + \frac{1}{2}\sum_{i=0}^{r-1}\left(4(r-1-i)\lambda^2\tau^2\right)\left(4\lambda^2\tau^2\right)\\
    &= 
    \frac{4}{3}r\lambda^3\tau^3 + 4r(r-1)\lambda^4\tau^4\\
    &=
    \frac{4\lambda^3t^3}{3r^2} + \frac{4(r-1)\lambda^4t^4}{r^3}.
\end{align}

Since $(r-1)/r\leq1$, it follows that
\begin{equation}
    \frac{1}{2}\left\Vert e^{t\mathcal{L}}-\mathcal{Q}^{(2)}_t\right\Vert_{\diamond}
    = O\!\left(\frac{\lambda^3t^3+\lambda^4t^4}{r^2}\right).
\end{equation}
If additionally $\lambda t\geq1$, this simplifies to
\begin{equation}
    \frac{1}{2}\left\Vert e^{t\mathcal{L}}-\mathcal{Q}^{(2)}_t\right\Vert_{\diamond}
    = O\!\left(\frac{\lambda^4t^4}{r^2}\right).
\end{equation}
Therefore, choosing
\begin{equation}
    r = O\!\left(\frac{\lambda^2t^2}{\sqrt{\varepsilon}}\right)
\end{equation}
is sufficient to ensure
\begin{equation}
    \frac{1}{2}\left\Vert e^{t\mathcal{L}}-\mathcal{Q}^{(2)}_t\right\Vert_{\diamond}
    \leq\varepsilon.
\end{equation}
This concludes the proof.
\end{proof}

\subsection{Implementation of the qSWIFT Map}
\label{sec:qswift-implementation}

Since the second-order qSWIFT map $\mathcal{Q}^{(2)}_t$ is trace-preserving but need not be completely positive, it cannot in general be implemented directly as a quantum channel. Nevertheless, in most applications, the primary interest lies in computing the expectation value of a given observable with respect to some quantum state. We also consider this setting in what follows. Specifically, we aim to compute the expectation value of an observable $G$ with respect to the state obtained by applying the exact quantum channel $e^{t\mathcal{L}}$ to an initial state $\rho$:
\begin{equation}
    g \coloneqq \operatorname{Tr}\!\left[Ge^{t\mathcal{L}}(\rho)\right].
\end{equation}
However, since $e^{t\mathcal{L}}$ cannot be implemented exactly in general, we instead estimate the second-order approximation to this expectation value,
\begin{equation}
    g^{(2)} \coloneqq \operatorname{Tr}\!\left[G\mathcal{Q}^{(2)}_t(\rho)\right].
\end{equation}
From Theorem~\ref{thm:qswift-map-error}, we have
\begin{equation}
\left | g-g^{(2)}\right|
\leq\left\Vert G\right\Vert\left\Vert e^{t\mathcal{L}}-\mathcal{Q}^{(2)}_t \right\Vert_{\diamond}
= O\!\left(\frac{\left\Vert G\right\Vert(\lambda^3t^3+\lambda^4t^4)}{r^2}\right).
\end{equation}
In particular, when $\lambda t\geq1$, this becomes
\begin{equation}
\left |g-g^{(2)}\right|= O\!\left(\frac{\left\Vert G\right\Vert\lambda^4t^4}{r^2}\right).
\end{equation}
Thus, to guarantee accuracy within a tolerance $\varepsilon > 0$, it suffices to choose
\begin{equation}
r=O\!\left(\frac{\lambda^2t^2\left\Vert G\right\Vert^{1/2}}{\varepsilon^{1/2}}\right),
\end{equation}
so that $|g - g^{(2)}|\leq \varepsilon$. Compared with the standard qDRIFT approximation, this gives a quadratic improvement in the dependence on the approximation error parameter $\varepsilon$ at the level of the number of time slices $r$.

In what follows, we demonstrate how to estimate $g^{(2)}$. From the definition of $\mathcal{Q}^{(2)}_t$ in~\eqref{def:qswift-map}, we have
\begin{align}
    g^{(2)} & = \operatorname{Tr}\!\left[G\mathcal{Q}_{\tau}^r(\rho)\right] + \frac{\lambda^2\tau^2}{2} \sum_{k=0}^{r-1}\operatorname{Tr}\!\left[G\mathcal{Q}_{\tau}^{r-1-k}\mathcal{L}^{(2)}\mathcal{Q}_{\tau}^{k}(\rho)\right]\\
    & \eqqcolon g^{(1)} + \delta g.
\end{align}
Here, $g^{(1)}\coloneqq \operatorname{Tr}\!\left[G\mathcal{Q}_{\tau}^r(\rho)\right]$ denotes the expectation value associated with the first-order qDRIFT approximation, while $\delta g\coloneqq \frac{\lambda^2\tau^2}{2} \sum_{k=0}^{r-1}\operatorname{Tr}\!\left[G\mathcal{Q}_{\tau}^{r-1-k}\mathcal{L}^{(2)}\mathcal{Q}_{\tau}^{k}(\rho)\right]$ denotes the second-order correction.
To estimate $g^{(1)}$, we apply the qDRIFT one-step procedure with time step $\tau$ repeatedly for $r$ steps, so that the averaged map is $\mathcal{Q}_{\tau}^{r}$. We then measure $G$, repeat this procedure multiple times, and take the sample mean. This yields an unbiased estimator of $g^{(1)}$. However, the estimation of the second term $\delta g$ is more involved and requires further decomposition.

To estimate $\delta g$, we express the correction as a sum over the position at which $\mathcal{L}^{(2)}$ is inserted:
\begin{equation}
    \delta g = \frac{\lambda^2\tau^2}{2} \sum_{k=0}^{r-1} \delta g_k\label{eq:g-g_k},
\end{equation}
where each term $\delta g_k$ is given by
\begin{equation}\label{eq:qswift_GQLQ}
    \delta g_k \coloneqq \operatorname{Tr}\!\left[ G\mathcal{Q}_{\tau}^{r-1-k}\mathcal{L}^{(2)}\mathcal{Q}_{\tau}^{k}(\rho) \right].
\end{equation}
From \eqref{eq:qswift-generator-approx} with $m=2$, we have
\begin{equation}
\mathcal{L}^{(2)} = \widetilde{\mathcal{L}}^2 - \sum_{\ell=1}^{J+K} p_{\ell}\mathcal{M}_{\ell}^2
= \sum_{\ell_1,\ell_2=1}^{J+K} p_{\ell_1}p_{\ell_2} \mathcal{M}_{\ell_1}\mathcal{M}_{\ell_2} - \sum_{\ell_1=1}^{J+K} p_{\ell_1}\mathcal{M}_{\ell_1}^2.
\end{equation}
By substituting this into \eqref{eq:qswift_GQLQ}, we obtain
\begin{equation}
\delta g_k = \sum_{\ell_1,\ell_2=1}^{J+K} p_{\ell_1}p_{\ell_2} \operatorname{Tr}\!\left[G\mathcal{Q}_{\tau}^{r-1-k} \mathcal{M}_{\ell_1}\mathcal{M}_{\ell_2}\mathcal{Q}_{\tau}^{k}(\rho)\right]
- \sum_{\ell_1=1}^{J+K} p_{\ell_1} \operatorname{Tr}\!\left[G\mathcal{Q}_{\tau}^{r-1-k}\mathcal{M}_{\ell_1}^2\mathcal{Q}_{\tau}^{k}(\rho)\right].\label{eq:expression-g_k}
\end{equation}
Note that both sums involve terms of a similar structure. The second sum can be viewed as the special case $\ell_2=\ell_1$. Thus, the basic building block is
\begin{equation}
\operatorname{Tr}\!\left[G\mathcal{Q}_{\tau}^{r-1-k}\mathcal{M}_{\ell_1}\mathcal{M}_{\ell_2}\mathcal{Q}_{\tau}^{k}(\rho)\right],
\label{eq:l_1-l_2-term}
\end{equation}
for given $\ell_1,\ell_2\in[J+K]$.

In order to estimate this building block, we employ a qSWIFT control qubit which we initialize in the state $|+\rangle \coloneqq (|0\rangle + |1\rangle)/\sqrt{2}$. 
The initial joint state of the control qubit and system qubits is therefore
\begin{equation}
    |+\rangle\!\langle+| \otimes \rho
    = \frac{1}{2}
    \begin{pmatrix}
        \rho & \rho \\
        \rho & \rho
    \end{pmatrix},
\end{equation}
where we use a block-matrix notation in which each block corresponds to an operator acting on the system Hilbert space.

We first apply the qDRIFT channel $\mathcal{Q}_{\tau}$ exactly $k$ times to the system (leaving the qSWIFT control qubit untouched). This yields
\begin{equation}
    |+\rangle\!\langle+| \otimes \mathcal{Q}_{\tau}^{k}(\rho)
    = \frac{1}{2}
    \begin{pmatrix}
        \mathcal{Q}_{\tau}^{k}(\rho) & \mathcal{Q}_{\tau}^{k}(\rho) \\
        \mathcal{Q}_{\tau}^{k}(\rho) & \mathcal{Q}_{\tau}^{k}(\rho)
    \end{pmatrix}.
\end{equation}

Before applying a particular map $\mathcal{M}_{\ell}$ (for some $\ell \in [J+K]$), we define \emph{swift maps} that allow us to implement $\mathcal{M}_{\ell}$ coherently on the off-diagonal blocks.

If $\mathcal{M}_{\ell}$ is a Hamiltonian term with Hamiltonian $H_{\ell}$, that is, $\mathcal{M}_{\ell}(\sigma) = -i[H_{\ell}, \sigma]$ for some $\ell \leq J$, then we define two swift maps $\mathcal{S}^{(0)}_{\ell}$ and $\mathcal{S}^{(1)}_{\ell}$, each acting on the off-diagonal block $\sigma$ as follows:
\begin{align}
    \begin{pmatrix}
        \cdot & \sigma \\
        \sigma & \cdot
    \end{pmatrix}
    &\xrightarrow{\mathcal{S}^{(0)}_{\ell}}
    \begin{pmatrix}
        \cdot & i\sigma H_{\ell} \\
        -i H_{\ell}\sigma & \cdot
    \end{pmatrix}, \\[0.3em]
    \begin{pmatrix}
        \cdot & \sigma \\
        \sigma & \cdot
    \end{pmatrix}
    &\xrightarrow{\mathcal{S}^{(1)}_{\ell}}
    \begin{pmatrix}
        \cdot & -i H_{\ell} \sigma \\
        i\sigma H_{\ell} & \cdot
    \end{pmatrix}.
\end{align}
Here the symbol ``$\cdot$'' indicates the diagonal blocks, whose transformation we do not display since only the off-diagonal terms are relevant for our construction. The full swift maps on arbitrary joint control-system operators will be defined below. It is clear that summing the above two transformations gives
\begin{equation}
    \begin{pmatrix}
        \cdot & \sigma \\
        \sigma & \cdot
    \end{pmatrix}
    \xrightarrow{\mathcal{S}^{(0)}_{\ell} + \mathcal{S}^{(1)}_{\ell}}
    \begin{pmatrix}
        \cdot & \mathcal{M}_{\ell}(\sigma) \\
        \mathcal{M}_{\ell}(\sigma) & \cdot
    \end{pmatrix}.
\end{equation}

If instead $\mathcal{M}_{\ell}$ is a Lindbladian term with Lindblad operator $L_{\ell-J}$, that is, $\mathcal{M}_{\ell}(\sigma) = L_{\ell-J}\sigma L^{\dagger}_{\ell-J} - \tfrac{1}{2}\{L_{\ell-J}^{\dagger}L_{\ell-J}, \sigma\}$ for some $J < \ell \leq J+K$, we employ the following three swift maps $\mathcal{S}^{(2)}_{\ell}$, $\mathcal{S}^{(3)}_{\ell}$, and $\mathcal{S}^{(4)}_{\ell}$:
\begin{align}
    \begin{pmatrix}
        \cdot & \sigma \\
        \sigma & \cdot
    \end{pmatrix}
    &\xrightarrow{\mathcal{S}^{(2)}_{\ell}}
    \begin{pmatrix}
        \cdot & L_{\ell-J} \sigma L_{\ell-J}^{\dagger} \\
        L_{\ell-J} \sigma L_{\ell-J}^{\dagger} & \cdot
    \end{pmatrix}, \\[0.3em]
    \begin{pmatrix}
        \cdot & \sigma \\
        \sigma & \cdot
    \end{pmatrix}
    &\xrightarrow{\mathcal{S}^{(3)}_{\ell}}
    \begin{pmatrix}
        \cdot & -\tfrac12 L_{\ell-J}^{\dagger} L_{\ell-J} \sigma \\
        -\tfrac12 \sigma L_{\ell-J}^{\dagger} L_{\ell-J} & \cdot
    \end{pmatrix}, \\[0.3em]
    \begin{pmatrix}
        \cdot & \sigma \\
        \sigma & \cdot
    \end{pmatrix}
    &\xrightarrow{\mathcal{S}^{(4)}_{\ell}}
    \begin{pmatrix}
        \cdot & -\tfrac12 \sigma L_{\ell-J}^{\dagger} L_{\ell-J} \\
        -\tfrac12 L_{\ell-J}^{\dagger} L_{\ell-J} \sigma & \cdot
    \end{pmatrix}.
\end{align}

The swift maps introduced above cannot in general be realized directly as deterministic unitary operations. 
In particular, for the Hamiltonian swift maps, the assumption $\Vert H_{\ell}\Vert=1$ does not imply that $H_{\ell}$ is unitary, while for the Lindbladian swift maps, the Lindblad operator $L_{\ell-J}$ need not be unitary.
However, under the locality and normalization assumptions introduced above, all of these swift maps can be realized as heralded local primitives using unitary dilations.
Indeed, since each $H_j$ and $L_k$ is local and satisfies $\Vert H_j\Vert=\Vert L_k\Vert=1$, the operators appearing in these maps remain local.
We show below that the corresponding block operators are contractions and therefore admit unitary dilations.

More explicitly, for system operators $A$ and $B$, define
\begin{equation}
    \mathsf{K}(A,B) \coloneqq |0\rangle\!\langle 0|\otimes A + |1\rangle\!\langle 1|\otimes B.
    \label{eq:swift-block-operator}
\end{equation}
The operator $\mathsf{K}(A,B)$ acts on the qSWIFT control qubit and the system.
For an arbitrary joint control-system operator written in block form as
\begin{equation}
    \Omega =
    \begin{pmatrix}
        \Omega_{00} & \Omega_{01} \\
        \Omega_{10} & \Omega_{11}
    \end{pmatrix},
\end{equation}
we have
\begin{equation}
    \mathsf{K}(A,B)\Omega \mathsf{K}(A,B)^\dagger =
    \begin{pmatrix}
        A\Omega_{00}A^\dagger & A\Omega_{01}B^\dagger \\
        B\Omega_{10}A^\dagger & B\Omega_{11}B^\dagger
    \end{pmatrix}.
\end{equation}

For a Hamiltonian term with $\ell\leq J$, define
\begin{equation}
    \mathsf{K}_{\ell}^{(0)} \coloneqq \mathsf{K}(I,-iH_{\ell}), \qquad
    \mathsf{K}_{\ell}^{(1)} \coloneqq \mathsf{K}(-iH_{\ell},I).
\end{equation}
For a Lindbladian term with $J<\ell\leq J+K$, define
\begin{align}
    \mathsf{K}_{\ell}^{(2)} &\coloneqq \mathsf{K}(L_{\ell-J},L_{\ell-J}), \\
    \mathsf{K}_{\ell}^{(3)} &\coloneqq \mathsf{K}(-L_{\ell-J}^\dagger L_{\ell-J}/2,I), \\
    \mathsf{K}_{\ell}^{(4)} &\coloneqq \mathsf{K}(I,-L_{\ell-J}^\dagger L_{\ell-J}/2).
\end{align}
The corresponding swift maps are defined on an arbitrary joint control-system operator $\Omega$ by
\begin{equation}
    \mathcal{S}_{\ell}^{(b)}(\Omega)
    \coloneqq
    \mathsf{K}_{\ell}^{(b)}\Omega \mathsf{K}_{\ell}^{(b)\dagger}.
\end{equation}
These definitions reproduce the off-diagonal actions displayed above.

To construct the corresponding unitary dilations, we first verify that the block operators defined above are contractions.
Since $\mathsf{K}(A,B)$ is block diagonal with respect to the qSWIFT control qubit, $\lVert \mathsf{K}(A,B)\rVert=\max\{\lVert A\rVert,\lVert B\rVert\}$.
Using $\Vert H_{\ell}\Vert=1$, $\Vert L_{\ell-J}\Vert=1$, and $\lVert L_{\ell-J}^\dagger L_{\ell-J}\rVert=1$, each of the operators $\mathsf{K}_{\ell}^{(b)}$ defined above satisfies $\lVert \mathsf{K}_{\ell}^{(b)}\rVert\leq1$.

For any contraction $C$, the Halmos unitary dilation~\cite{Levy2014DilationTheory} is given by
\begin{equation}
    U_C
    \coloneqq
    \begin{pmatrix}
        C & \sqrt{I-CC^\dagger} \\
        \sqrt{I-C^\dagger C} & -C^\dagger
    \end{pmatrix}.
    \label{eq:unitary-dilation-swift}
\end{equation}
By construction, $(\langle 0|\otimes I)U_C(|0\rangle\otimes I)=C$. 
Taking $C=\mathsf{K}_{\ell}^{(b)}$ therefore gives a unitary acting on one additional heralding qubit together with the qSWIFT control qubit and the local system register.
The resulting finite-dimensional unitary can in principle be decomposed into a sequence of two-level unitaries and implemented as a quantum circuit~\cite{Hu2020QuantumAlgorithmEvolving}.

Operationally, for either a Hamiltonian or Lindbladian swift map, we initialize the heralding qubit in the state $|0\rangle$, apply the corresponding unitary dilation, and measure the heralding qubit.
If the heralding measurement fails, the contribution of that shot is recorded as zero and retained in the sample mean.
Averaging over the heralding outcomes then reproduces the unnormalized transformation associated with $\mathsf{K}_{\ell}^{(b)}$, so no additional normalization by the success probability is required.

Summing the maps $\mathcal{S}^{(2)}_{\ell}$, $\mathcal{S}^{(3)}_{\ell}$, and $\mathcal{S}^{(4)}_{\ell}$ yields
\begin{equation}
    \begin{pmatrix}
        \cdot & \sigma \\
        \sigma & \cdot
    \end{pmatrix}
    \xrightarrow{\mathcal{S}^{(2)}_{\ell} + \mathcal{S}^{(3)}_{\ell} + \mathcal{S}^{(4)}_{\ell}}
    \begin{pmatrix}
        \cdot & \mathcal{M}_{\ell}(\sigma) \\
        \mathcal{M}_{\ell}(\sigma) & \cdot
    \end{pmatrix}.
\end{equation}

After applying the two swift-map primitives, we apply $\mathcal{Q}_{\tau}^{r-1-k}$ to the system and measure $X\otimes G$, with $X$ acting on the qSWIFT control qubit and $G$ acting on the system. 
If any heralding measurement fails, the outcome of that shot is assigned the value zero. 
With this convention, the contributions of the swift maps can be combined to reproduce the building block in~\eqref{eq:l_1-l_2-term}.

Now to rewrite the building block in~\eqref{eq:l_1-l_2-term}, we define the following set:
\begin{equation}
    \mathcal{B}_{\ell} \coloneqq \begin{cases}
    \{0,1\}, & \ell \le J \quad \text{(Hamiltonian term)}, \\[0.3em]
    \{2,3,4\}, & J < \ell \le J+K \quad \text{(Lindbladian term)}.
\end{cases}\label{eq:definition-which-term}
\end{equation}
Using this definition, the building block can be rewritten as follows:
\begin{multline}
    \operatorname{Tr}\!\left[ G\mathcal{Q}_{\tau}^{r-1-k}\mathcal{M}_{\ell_1} \mathcal{M}_{\ell_2}\mathcal{Q}_{\tau}^{k}(\rho) \right]\\
    = \sum_{\substack{b_1 \in \mathcal{B}_{\ell_1}\\b_2 \in \mathcal{B}_{\ell_2}}
    }\operatorname{Tr}\!\left[( X\otimes G)\left( \mathrm{id} \otimes \mathcal{Q}_{\tau}^{r-1-k}\right)\mathcal{S}_{\ell_1}^{(b_1)} \mathcal{S}_{\ell_2}^{(b_2)}\left( \mathrm{id} \otimes \mathcal{Q}_{\tau}^{k}\right)\left(|+\rangle\!\langle+|\otimes \rho\right) \right].\label{eq:target-quanti-alternate-exp}
\end{multline}

Finally, using~\eqref{eq:g-g_k},~\eqref{eq:expression-g_k}, and the expression of the building block in~\eqref{eq:target-quanti-alternate-exp}, we can rewrite $\delta g$ as follows:
\begin{align}
    \delta g & = \frac{\lambda^2\tau^2}{2} \sum_{k=0}^{r-1} \Bigg(\sum_{\ell_1, \ell_2=1}^{J+K}p_{\ell_1}p_{\ell_2}\sum_{\substack{b_1 \in \mathcal{B}_{\ell_1}\\b_2 \in \mathcal{B}_{\ell_2}}
    }\operatorname{Tr}\!\left[( X\otimes G)\left( \mathrm{id} \otimes \mathcal{Q}_{\tau}^{r-1-k}\right)\mathcal{S}_{\ell_1}^{(b_1)} \mathcal{S}_{\ell_2}^{(b_2)}\left( \mathrm{id} \otimes \mathcal{Q}_{\tau}^{k}\right)\left(|+\rangle\!\langle+|\otimes \rho\right) \right] \nonumber\\
    & \qquad - \sum_{\ell_1=1}^{J+K}p_{\ell_1}\sum_{b_1,b_2\in\mathcal{B}_{\ell_1}}\operatorname{Tr}\!\left[( X\otimes G)\left( \mathrm{id} \otimes \mathcal{Q}_{\tau}^{r-1-k}\right)\mathcal{S}_{\ell_1}^{(b_1)} \mathcal{S}_{\ell_1}^{(b_2)}\left( \mathrm{id} \otimes \mathcal{Q}_{\tau}^{k}\right)\left(|+\rangle\!\langle+|\otimes \rho\right) \right]\Bigg)\\
    & = \frac{\lambda^2 t^2}{r}\sum_{s=0}^{1}\sum_{k=0}^{r-1}\sum_{\ell_1, \ell_2=1}^{J+K} \sum_{\substack{b_1 \in \mathcal{B}_{\ell_1}\\b_2 \in \mathcal{B}_{\ell_2}}
    } P(s, k, \ell_1, \ell_2, b_1, b_2)\nonumber\\ 
    & \qquad \times \frac{(-1)^s}{p_{b_1}^{(\ell_1)} p_{b_2}^{(\ell_2)}} \operatorname{Tr}\!\left[( X\otimes G)\left( \mathrm{id} \otimes \mathcal{Q}_{\tau}^{r-1-k}\right)\mathcal{S}_{\ell_1}^{(b_1)} \mathcal{S}_{\ell_2}^{(b_2)}\left( \mathrm{id} \otimes \mathcal{Q}_{\tau}^{k}\right)\left(|+\rangle\!\langle+|\otimes \rho\right) \right]\label{eq:delta-g-final-expression},
\end{align}
where the probability distribution $ P(s, k, \ell_1, \ell_2, b_1, b_2)$ is defined as
\begin{equation}\label{eq:pd-qswift}
    P(s, k, \ell_1, \ell_2, b_1, b_2) \coloneqq \begin{cases}
\frac{1}{2} \cdot \frac{1}{r} p_{\ell_1} p_{\ell_2} p_{b_1}^{(\ell_1)} p_{b_2}^{(\ell_2)}, & s = 0, \\[0.3em]
\frac{1}{2} \cdot \frac{1}{r} p_{\ell_1} \delta_{\ell_1,\ell_2}p_{b_1}^{(\ell_1)} p_{b_2}^{(\ell_1)}, & s = 1.
\end{cases}
\end{equation}
Here the probabilities $ p_{b}^{(\ell)}$ are defined by $ p_{b}^{(\ell)} \coloneqq \tfrac{1}{|\mathcal{B}_{\ell}|}$.

From~\eqref{eq:delta-g-final-expression}, it is clear that one can use sampling to estimate $\delta g$. In the following, we present an algorithm whose output is an unbiased estimator of $\delta g$.\\

\noindent
\textbf{Unbiased Estimator for $\delta g$.}  
To estimate $\delta g$, repeat the following procedure independently several times and take the sample mean of the recorded outcomes:
\begin{enumerate}
    \item Sample $s\in\{0,1\}$ uniformly and sample $k\in\{0,\ldots,r-1\}$ uniformly.
    \item If $s=0$, sample $\ell_1$ and $\ell_2$ independently from the distribution $\{p_\ell\}_{\ell=1}^{J+K}$. 
    If $s=1$, sample~$\ell_1$ from $\{p_\ell\}_{\ell=1}^{J+K}$ and set $\ell_2=\ell_1$.
    \item Sample $b_1\in\mathcal{B}_{\ell_1}$ and $b_2\in\mathcal{B}_{\ell_2}$ uniformly. 
    Denote the corresponding probabilities by $p_{b_1}^{(\ell_1)}$ and $p_{b_2}^{(\ell_2)}$, where $p_b^{(\ell)}=1/|\mathcal{B}_\ell|$.
    \item Prepare the joint initial state $|+\rangle\!\langle+| \otimes \rho$, where the qSWIFT control qubit is initialized in $|+\rangle$ and the system in $\rho$.
    \item Apply the qDRIFT one-step procedure $k$ times to the system, so that the averaged map is $\mathcal{Q}_{\tau}^{k}$.
    \item Apply $\mathcal{S}_{\ell_2}^{(b_2)}$ first and then $\mathcal{S}_{\ell_1}^{(b_1)}$, so that the overall composition is $\mathcal{S}_{\ell_1}^{(b_1)}\mathcal{S}_{\ell_2}^{(b_2)}$. 
    For each Hamiltonian or Lindbladian swift map, initialize the heralding qubit in $|0\rangle$, apply the corresponding local unitary dilation, and measure the heralding qubit.
    \item If any heralding measurement fails, record $0$ for this shot. 
    Otherwise, apply the qDRIFT one-step procedure $r-1-k$ more times to the system, so that the averaged map is $\mathcal{Q}_{\tau}^{r-1-k}$, and measure $X\otimes G$, with $X$ acting on the qSWIFT control qubit and $G$ acting on the system.
    \item For a successful shot, multiply the measurement outcome by $(-1)^s/(p_{b_1}^{(\ell_1)}p_{b_2}^{(\ell_2)})$, scale it by $\lambda^2t^2/r$, and record the result.
\end{enumerate}
The expectation value of the recorded outcome is $\delta g$, so the sample mean gives an unbiased estimator of $\delta g$.

\subsection{Sketch of Higher-Order qSWIFT-Type Corrections}

The second-order construction above uses only the leading correction term in the expansion~\eqref{eq:qswift-correction}. 
The same idea can be extended formally to higher orders.
For $m\geq 2$, define the linear one-step correction map
\begin{equation}\label{eq:delta-m-definition}
    \Delta_m \coloneqq\frac{\lambda^m\tau^m}{m!}\mathcal{L}^{(m)}.
\end{equation}
Then \eqref{eq:qswift-correction} can be rewritten as
\begin{equation}\label{eq:qswift-delta-expansion}
    e^{\tau\mathcal{L}} = \mathcal{Q}_{\tau} + \sum_{m=2}^{\infty} \Delta_m.
\end{equation}
Expanding the $r$-fold product of this one-step expression produces terms containing insertions of the correction maps $\Delta_m$ among the averaged qDRIFT steps. 
The second-order construction keeps the baseline term $\mathcal{Q}_{\tau}^{r}$ and all terms with a single insertion of $\Delta_2$, which gives~\eqref{def:qswift-map}.

To go one order further, the next correction would include not only all terms with a single insertion of $\Delta_3$, but also the terms with two insertions of $\Delta_2$. 
Indeed, after summing over possible insertion positions, both types of terms contribute at the same formal order in $1/r$. 
Higher-order qSWIFT-type corrections can be constructed in the same way by expanding the $r$-fold product and retaining all correction-insertion patterns that contribute at the desired order.

This order-counting also explains the expected improvement in the dependence on the target precision. 
The first-order qDRIFT-type approximation leaves an $O(r^{-1})$ error, leading to an $\varepsilon^{-1}$ dependence in the number of time slices. 
The second-order construction matches the $r$-fold expansion through order $O(r^{-1})$ and yields an $O(r^{-2})$ error, leading to the $\varepsilon^{-1/2}$ dependence proved above. 
Formally, a $p$-th order qSWIFT-type construction that matches the $r$-fold expansion through order $O(r^{-(p-1)})$ would be expected to have a residual error of order $O(r^{-p})$, and hence an $\varepsilon^{-1/p}$ dependence in the number of time slices, up to problem-dependent factors in $\lambda t$.

This discussion is only a formal sketch. 
A rigorous higher-order theorem would require a separate diamond-norm error analysis for the resulting remainder terms. 
We leave a complete analysis of the full higher-order hierarchy for future work.

\section{Time-Dependent Lindbladian Simulation}

\subsection{Problem Statement}

The time-dependent Lindblad master equation is given as follows:
\begin{equation}
    \frac{d \rho(t)}{dt} = \mathcal{L}(t)(\rho(t)) \coloneqq -i[H(t), \rho(t)] + \sum_{k=1}^{K} \left[L_k(t)\rho(t)L_k^\dagger(t) - \frac{1}{2} \{L_k^\dagger(t)L_k(t), \rho(t)\}\right] \label{eq:time-dep-Lind},
\end{equation}
where $H(t)$ is a time-dependent Hamiltonian and $L_1(t), \ldots, L_K(t)$ are time-dependent Lindblad operators, which are not Hermitian in general. Unlike the time-independent case, the solution generally requires time ordering. Its formal solution can be expressed as
\begin{equation}
    \rho(t) = \operatorname{exp}_{\mathcal{T}}\!\left(\int_0^t d\tau\ \mathcal{L}(\tau) \right)\!(\rho(0)),
\end{equation}
where $\operatorname{exp}_{\mathcal{T}}$ is the time-ordered exponential map. For conciseness, we introduce the shorthand notation:
\begin{equation}
    \mathcal{E}(t, v) \coloneqq \operatorname{exp}_{\mathcal{T}}\!\left(\int_v^t d\tau\ \mathcal{L}(\tau) \right),\label{eq:target-qc}
\end{equation}
with $\mathcal{E}(v,v)=\mathcal{I}$.
This allows us to compactly write:
\begin{equation}
    \rho(t) =  \mathcal{E}(t, 0)\!(\rho(0)).
\end{equation}
Furthermore, we can also rewrite the Lindblad master equation~\eqref{eq:time-dep-Lind} in the quantum channel form as
\begin{equation}
    \frac{d \mathcal{E}(t, 0)}{dt} = \mathcal{L}(t) \circ \mathcal{E}(t, 0).
    \label{eq:rewrite-TD-Lindblad}
\end{equation}
We assume that $\mathcal{L}(\tau)$ is bounded and piecewise continuous on $[0,t]$ and nonzero everywhere.
Under this assumption, the differential equation above holds almost everywhere in time.
Our aim is to approximately simulate the quantum channel $\mathcal{E}(t, 0)$ on a quantum computer with an approximation error of $\varepsilon > 0$.

\subsection{Technical Lemma}

Before moving forward with the algorithm, we present the following technical lemma, which we will use later for the error analysis. This lemma is concerned with the scaled version of the map $\mathcal{E}(t, v)$, defined as
\begin{equation}
    \mathcal{E}_s(t, v) \coloneqq \operatorname{exp}_{\mathcal{T}}\!\left(\int_v^t d\tau\ s\mathcal{L}(\tau) \right)\label{eq:scaled-E-channel},
\end{equation}
where $s$ is a scalar parameter.
Specifically, in this lemma, we analyze the derivative of $\mathcal{E}_s(t, v)$ with respect to $s$:
\begin{equation}
    \frac{d}{ds}\mathcal{E}_s(t, v). 
\end{equation}

\begin{lemma} \label{lem:time-dep-deriv-for-err-bnd}
The following holds:
\begin{equation}
        \frac{d \mathcal{E}_s(t, v)}{ds} = \int_v^t d\tau\ \mathcal{E}_s(t, \tau) \circ \mathcal{L}(\tau) \circ \mathcal{E}_s(\tau, v).
\end{equation}
\end{lemma}

\begin{proof}
First, we obtain the integral form of $\mathcal{E}_s(t, v)$ using the fundamental theorem of calculus. Since $\mathcal{L}(\tau)$ is bounded and piecewise continuous, the differential equation in physical time holds almost everywhere, while the corresponding integral identity is valid for all endpoints. Thus,
\begin{equation}
    \mathcal{E}_s(t, v) - \mathcal{E}_s(v, v) = \int_v^t d\tau\, \frac{d \mathcal{E}_s(\tau, v)}{d\tau} = s \int_v^t d\tau\, \mathcal{L}(\tau) \circ \mathcal{E}_s(\tau, v).
\end{equation}
Since $\mathcal{E}_s(v, v) = \mathcal{I}$, where $\mathcal{I}$ is the identity channel, we obtain
\begin{equation}
   \mathcal{E}_s(t, v) = \mathcal{I} + s \int_v^t d\tau\, \mathcal{L}(\tau) \circ \mathcal{E}_s(\tau, v).
\end{equation}
The Dyson series for $\mathcal{E}_s(t,v)$ converges uniformly for $s$ in compact intervals, which justifies differentiation with respect to $s$.
Taking the derivative with respect to $s$ on both sides of the above equation, we get
\begin{equation}
    \frac{d \mathcal{E}_s(t, v)}{ds} = \int_v^t d\tau\, \mathcal{L}(\tau) \circ \mathcal{E}_s(\tau, v) + s \int_v^t d\tau\, \mathcal{L}(\tau) \circ \frac{d \mathcal{E}_s(\tau, v)}{ds} \label{eq:deriv-E-s}.
\end{equation}
Let
\begin{equation}
    Y(t)\coloneqq \frac{d}{ds}\mathcal{E}_s(t,v).
\end{equation}
Differentiating \eqref{eq:deriv-E-s} with respect to $t$ for almost every $t$, we obtain
\begin{equation}
    \frac{dY(t)}{dt} = \mathcal{L}(t)\circ\mathcal{E}_s(t,v) + s\mathcal{L}(t)\circ Y(t).
\end{equation}
Moreover, since $\mathcal{E}_s(v,v)=\mathcal{I}$ is independent of $s$, we have
\begin{equation}
    Y(v) = \frac{d}{ds}\mathcal{E}_s(v,v) = 0.
\end{equation}
Applying the variation-of-parameters formula to this inhomogeneous equation yields
\begin{equation}
    Y(t) = \int_v^t d\tau\, \mathcal{E}_s(t,\tau) \circ \mathcal{L}(\tau) \circ \mathcal{E}_s(\tau,v).
\end{equation}
Substituting back $Y(t)=\tfrac{d}{ds}\mathcal{E}_s(t,v)$ gives
\begin{equation}
    \frac{d \mathcal{E}_s(t,v)}{ds} = \int_v^t d\tau\, \mathcal{E}_s(t,\tau)\circ \mathcal{L}(\tau)\circ\mathcal{E}_s(\tau,v),
\end{equation}
as claimed.
\end{proof}

\subsection{First-Order Sampling Algorithm}\label{sec:first-order-sampling-algorithm}

Recall that our target is to simulate the quantum channel $\mathcal{E}(t, 0)$ defined in~\eqref{eq:target-qc}. 
Following the continuous qDRIFT approach for time-dependent Hamiltonian simulation in~\cite{Berry2020TimeDependentHamiltonian}, we extend the same norm-weighted time-sampling idea to Lindbladian dynamics.
On a single interval $[v,t]$, the target propagator satisfies the first-order expansion
\begin{equation}
    \mathcal{E}(t,v) = \mathcal{I} + \int_v^t d\tau\,\mathcal{L}(\tau) + \text{higher-order terms}.
\end{equation}
Motivated by this expansion, we define a randomized one-step channel whose average matches the first-order term.
Specifically, we sample a time $\tau\in[v,t]$ from a probability density $p(\tau)$ and apply the channel $e^{\mathcal{L}(\tau)/p(\tau)}$.
The resulting averaged channel is
\begin{equation}
    \mathcal{A}(t, v) \coloneqq \int_{v}^{t} d\tau\ p(\tau)\ e^{\frac{\mathcal{L}(\tau)}{p(\tau)}}\label{eq:rand-qc},
\end{equation}
where $p(\tau)$ is a probability density function on $\tau \in [v, t]$ given by
\begin{align}
    p(\tau) & \coloneqq \frac{\left \Vert \mathcal{L}(\tau) \right \Vert_{\diamond}}{\left \Vert \mathcal{L} \right \Vert_{\diamond, 1}}\label{eq:prob-den}, \\
    \left \Vert \mathcal{L} \right \Vert_{\diamond, 1} & \coloneqq \int_v^t d\tau \left \Vert \mathcal{L}(\tau) \right \Vert_{\diamond}.
\end{align}
To simulate the full evolution on $[0,t]$, we compose the one-step averaged channels over a partition $0=t_0<t_1<\cdots<t_{r-1}<t_r=t$.
The full averaged sampler is
\begin{equation}
    \mathcal{A}(t_r,t_{r-1}) \circ \cdots \circ \mathcal{A}(t_2,t_1) \circ \mathcal{A}(t_1,t_0).
\end{equation}
Equivalently, with products of maps ordered from right to left, we write
\begin{align}
    \prod_{i=0}^{r-1} \mathcal{A}(t_{i+1}, t_{i})&\coloneqq\mathcal{A}(t_r,t_{r-1}) \circ \cdots \circ \mathcal{A}(t_2,t_1) \circ \mathcal{A}(t_1,t_0)\notag\\ &=  \prod_{i=0}^{r-1}\left[ \int_{t_i}^{t_{i+1}} d\tau\ p^{(i)}(\tau)\ e^{\frac{\mathcal{L}(\tau)}{p^{(i)}(\tau)}}\right]\label{eq:comp-protocol},
\end{align}
where each $p^{(i)}(\tau)$ is a probability density function defined analogously to~\eqref{eq:prob-den} over the interval $[t_i, t_{i+1}]$:
\begin{equation}
    p^{(i)}(\tau) \coloneqq \frac{\left \Vert \mathcal{L}(\tau) \right \Vert_{\diamond}}{\left \Vert \mathcal{L} \right \Vert_{\diamond, 1}^{(i)}},\label{eq:prob-den-i}
\end{equation}
with
\begin{equation}
    \left \Vert \mathcal{L} \right \Vert_{\diamond, 1}^{(i)} \coloneqq \int_{t_i}^{t_{i+1}} d\tau\ \left \Vert \mathcal{L}(\tau) \right \Vert_{\diamond}. 
\end{equation}
We assume that each density $p^{(i)}$ can be sampled efficiently. 
We choose the partition so that each interval has the same integrated Lindbladian norm:
\begin{equation}
    \left \Vert \mathcal{L} \right \Vert_{\diamond, 1}^{(i)} = \int_{t_i}^{t_{i+1}} d\tau\ \left \Vert \mathcal{L}(\tau) \right \Vert_{\diamond} =  \frac{1}{r}\int_{0}^{t} d\tau\ \left \Vert \mathcal{L}(\tau) \right \Vert_{\diamond} = \frac{1}{r}\left \Vert \mathcal{L} \right \Vert_{\diamond, 1}.\label{eq:partition}
\end{equation}
Such a partition exists by applying the intermediate value theorem to the cumulative function $F(x)=\int_0^x d\tau\,\Vert\mathcal{L}(\tau)\Vert_{\diamond}$.

We now analyze the approximation error by bounding the diamond distance between the target channel $\mathcal{E}(t, 0)$ and the sampler $ \prod_{i=0}^{r-1} \mathcal{A}(t_{i+1}, t_{i})$. 
\begin{theorem}\label{thm:error-analysis-gen}
    Let $\mathcal{L}(\tau)$ be a time-dependent Lindbladian as defined in~\eqref{eq:time-dep-Lind} for all $\tau \in [0, t]$. We assume that it is bounded and piecewise continuous on $[0,t]$, and it is nonzero everywhere. Let $\mathcal{E}(t, 0) =  \operatorname{exp}_{\mathcal{T}}\!\left(\int_0^t d\tau\ \mathcal{L}(\tau) \right)$ be the solution to the corresponding Lindblad master equation, and let $\mathcal{A}(t, v)$ be the quantum channel defined in~\eqref{eq:rand-qc}. Let $0=t_0<t_1<\cdots<t_r=t$ be a partition satisfying~\eqref{eq:partition}. Then, for every positive integer $r$, the following holds:
    \begin{equation}
        \left \Vert \mathcal{E}(t, 0) - \prod_{i=0}^{r-1} \mathcal{A}(t_{i+1}, t_{i}) \right \Vert_{\diamond} \leq \frac{\left \Vert \mathcal{L} \right \Vert_{\diamond, 1}^2}{r}. 
    \end{equation}
\end{theorem}

\begin{proof}
    Using the composition property of the time-ordered propagator $\mathcal{E}(t,v)=\mathcal{E}(t,u)\circ\mathcal{E}(u,v)$ \cite[Eq.~(3.77)]{breuer2002TheoryOpenQuantum}, together with a telescoping expansion and the fact that quantum channels have diamond norm one, we obtain
    \begin{align}
        \left \Vert \mathcal{E}(t, 0) - \prod_{i=0}^{r-1} \mathcal{A}(t_{i+1}, t_{i}) \right \Vert_{\diamond} \leq \sum_{i=0}^{r-1} \left \Vert \mathcal{E}(t_{i+1}, t_{i}) - \mathcal{A}(t_{i+1}, t_{i}) \right \Vert_{\diamond}\label{eq:sum-split-error}.
    \end{align}
    
    Now, for each $i\in\{0,\ldots,r-1\}$ and $s\in[0,1]$, let us define a scaled version of the quantum channel $\mathcal{A}(t_{i+1}, t_{i})$:
    \begin{equation}
        \mathcal{A}_s(t_{i+1}, t_{i}) \coloneqq \int_{t_i}^{t_{i+1}} d\tau\ p^{(i)}(\tau)\ e^{s\frac{\mathcal{L}(\tau)}{p^{(i)}(\tau)}}\label{eq:scaled-rand-qc}.
    \end{equation}
    This implies that $ \mathcal{A}_1(t_{i+1}, t_{i}) =  \mathcal{A}(t_{i+1}, t_{i})$ and $\mathcal{A}_0(t_{i+1}, t_{i}) = \mathcal{I}$. Similarly, using the scaled version of the channel $\mathcal{E}(t, v)$, as defined in~\eqref{eq:scaled-E-channel}, for all $i\in\{0,\ldots,r-1\}$, we have $\mathcal{E}_1(t_{i+1},t_i)=\mathcal{E}(t_{i+1},t_i)$ and $\mathcal{E}_0(t_{i+1},t_i)=\mathcal{I}$.

    Now consider the following:
    \begin{align}
        \left \Vert \mathcal{E}(t_{i+1}, t_{i}) - \mathcal{A}(t_{i+1}, t_{i}) \right \Vert_{\diamond} & = \left \Vert \mathcal{E}_1(t_{i+1}, t_{i}) - \mathcal{A}_1(t_{i+1}, t_{i}) \right \Vert_{\diamond} \label{eq:use-1}\\
        & =  \left \Vert \mathcal{E}_1(t_{i+1}, t_{i}) - \mathcal{E}_0(t_{i+1}, t_{i}) + \mathcal{A}_0(t_{i+1}, t_{i}) - \mathcal{A}_1(t_{i+1}, t_{i}) \right \Vert_{\diamond}\\
        &  =  \left \Vert \mathcal{E}_1(t_{i+1}, t_{i}) - \mathcal{A}_1(t_{i+1}, t_{i}) - \left (\mathcal{E}_0(t_{i+1}, t_{i}) - \mathcal{A}_0(t_{i+1}, t_{i})\right) \right \Vert_{\diamond},
    \end{align}
    where the second equality follows due to the fact that $\mathcal{E}_0(t_{i+1}, t_{i}) = \mathcal{A}_0(t_{i+1}, t_{i}) = \mathcal{I}$.
    Using the fundamental theorem of calculus, consider that
    \begin{align}
         \left \Vert \mathcal{E}(t_{i+1}, t_{i}) - \mathcal{A}(t_{i+1}, t_{i}) \right \Vert_{\diamond} &  =  \left \Vert \int_0^1 ds_1\ \frac{d}{ds_1} \left(\mathcal{E}_{s_1}(t_{i+1}, t_{i}) - \mathcal{A}_{s_1}(t_{i+1}, t_{i}) \right) \right \Vert_{\diamond}\label{eq:first-fun-cal}.
    \end{align}
    First, by Lemma~\ref{lem:time-dep-deriv-for-err-bnd},
    \begin{align}
        \left.\frac{d}{ds_1} \mathcal{E}_{s_1}(t_{i+1}, t_{i})\right |_{s_1 = 0} & = \left. \int_{t_{i}}^{t_{i+1}} d\tau\, \mathcal{E}_{s_1}(t_{i+1}, \tau) \circ \mathcal{L}(\tau) \circ \mathcal{E}_{s_1}(\tau, t_{i}) \right |_{s_1 = 0}\\
        & = \int_{t_{i}}^{t_{i+1}} d\tau\, \mathcal{L}(\tau)\label{eq:first-order-E-identity}.
    \end{align}
    On the other hand, direct differentiation of $\mathcal{A}_{s_1}$ gives
    \begin{align}
         \left.\frac{d}{ds_1} \mathcal{A}_{s_1}(t_{i+1}, t_{i})\right |_{s_1 = 0} & = \left.\frac{d}{ds_1} \int_{t_i}^{t_{i+1}} d\tau\, p^{(i)}(\tau)\ e^{s_1\frac{\mathcal{L}(\tau)}{p^{(i)}(\tau)}}\right|_{s_1 = 0}\\
         & = \left.\int_{t_i}^{t_{i+1}} d\tau\, p^{(i)}(\tau)\ \frac{\mathcal{L}(\tau)}{p^{(i)}(\tau)}\, e^{s_1\frac{\mathcal{L}(\tau)}{p^{(i)}(\tau)}}\right |_{s_1 = 0}\\
         & = \int_{t_{i}}^{t_{i+1}} d\tau\, \mathcal{L}(\tau)\label{eq:first-order-A-identity}.
    \end{align}
    Since the two first-order terms in~\eqref{eq:first-order-E-identity} and~\eqref{eq:first-order-A-identity} are equal, we can add and subtract this common term inside the integrand in~\eqref{eq:first-fun-cal}. Thus,
    \begin{align}
        & \left \Vert \mathcal{E}(t_{i+1}, t_{i}) - \mathcal{A}(t_{i+1}, t_{i}) \right \Vert_{\diamond} \notag\\ 
        &=  \left \Vert \int_0^1 ds_1\, \left (\frac{d}{ds_1} \mathcal{E}_{s_1}(t_{i+1}, t_{i}) - \left.\frac{d}{ds_1} \mathcal{E}_{s_1}(t_{i+1}, t_{i})\right |_{s_1 = 0} +  \left.\frac{d}{ds_1} \mathcal{A}_{s_1}(t_{i+1}, t_{i})\right |_{s_1 = 0} - \frac{d}{ds_1} \mathcal{A}_{s_1}(t_{i+1}, t_{i})\right) \right \Vert_{\diamond}\\
        &=  \Bigg \Vert \int_0^1 ds_1\ \Bigg(\frac{d}{ds_1} \mathcal{E}_{s_1}(t_{i+1}, t_{i}) - \frac{d}{ds_1} \mathcal{A}_{s_1}(t_{i+1}, t_{i}) \notag\\
        & \hspace{6.5cm}- \left (\left.\frac{d}{ds_1} \mathcal{E}_{s_1}(t_{i+1}, t_{i})\right |_{s_1 = 0} -  \left.\frac{d}{ds_1} \mathcal{A}_{s_1}(t_{i+1}, t_{i})\right |_{s_1 = 0}\right) \Bigg) \Bigg \Vert_{\diamond}.
    \end{align}
    Applying the fundamental theorem of calculus again, we get
    \begin{align}
        \left \Vert \mathcal{E}(t_{i+1}, t_{i}) - \mathcal{A}(t_{i+1}, t_{i}) \right \Vert_{\diamond}
        & = \left \Vert \int_0^1 ds_1\,  \int_{0}^{s_1} ds_2\, \frac{d^2}{ds_2^2} \left(\mathcal{E}_{s_2}(t_{i+1}, t_{i}) - \mathcal{A}_{s_2}(t_{i+1}, t_{i}) \right) \right \Vert_{\diamond}\\
        & \leq  \int_0^1 ds_1\,  \int_{0}^{s_1} ds_2\, \left(\left \Vert \frac{d^2}{ds_2^2} \mathcal{E}_{s_2}(t_{i+1}, t_{i}) \right \Vert_{\diamond}  + \left \Vert \frac{d^2}{ds_2^2} \mathcal{A}_{s_2}(t_{i+1}, t_{i}) \right \Vert_{\diamond}\right)\label{eq:second-order}.
    \end{align}

    Now let us bound these terms one by one:
    \begin{align}
        & \left \Vert \frac{d^2}{ds_2^2} \mathcal{E}_{s_2}(t_{i+1}, t_{i}) \right \Vert_{\diamond} \notag\\
        & = \left \Vert \frac{d}{ds_2} \left ( \int_{t_{i}}^{t_{i+1}} d\tau\ \mathcal{E}_{s_2}(t_{i+1}, \tau) \circ \mathcal{L}(\tau) \circ \mathcal{E}_{s_2}(\tau, t_{i})  \right) \right \Vert_{\diamond}\label{eq:bound-first-term}\\
        & = \left \Vert\int_{t_{i}}^{t_{i+1}} d\tau\ \left(\frac{d}{ds_2} \mathcal{E}_{s_2}(t_{i+1}, \tau) \circ \mathcal{L}(\tau) \circ \mathcal{E}_{s_2}(\tau, t_{i})  +  \mathcal{E}_{s_2}(t_{i+1}, \tau) \circ \mathcal{L}(\tau) \circ \frac{d}{ds_2} \mathcal{E}_{s_2}(\tau, t_{i}) \right)\right \Vert_{\diamond}\\
        & \leq  \int_{t_{i}}^{t_{i+1}} d\tau\ \left(\left \Vert \frac{d}{ds_2} \mathcal{E}_{s_2}(t_{i+1}, \tau) \circ \mathcal{L}(\tau) \circ \mathcal{E}_{s_2}(\tau, t_{i}) \right \Vert_{\diamond}  +  \left \Vert \mathcal{E}_{s_2}(t_{i+1}, \tau) \circ \mathcal{L}(\tau) \circ \frac{d}{ds_2} \mathcal{E}_{s_2}(\tau, t_{i}) \right \Vert_{\diamond}\right)\\
         & \leq  \int_{t_{i}}^{t_{i+1}} d\tau\ \Bigg(\left \Vert \frac{d}{ds_2} \mathcal{E}_{s_2}(t_{i+1}, \tau) \right \Vert_{\diamond} \left \Vert \mathcal{L}(\tau) \right \Vert_{\diamond} \left \Vert \mathcal{E}_{s_2}(\tau, t_{i}) \right \Vert_{\diamond}  \notag \\
         & \hspace{6cm} +  \left \Vert \mathcal{E}_{s_2}(t_{i+1}, \tau) \right \Vert_{\diamond} \left \Vert \mathcal{L}(\tau)\right \Vert_{\diamond} \left \Vert \frac{d}{ds_2} \mathcal{E}_{s_2}(\tau, t_{i}) \right \Vert_{\diamond}\Bigg)\\
         & = \int_{t_{i}}^{t_{i+1}} d\tau\ \left \Vert \mathcal{L}(\tau) \right \Vert_{\diamond} \Bigg(\left \Vert \frac{d}{ds_2} \mathcal{E}_{s_2}(t_{i+1}, \tau) \right \Vert_{\diamond} + \left \Vert \frac{d}{ds_2} \mathcal{E}_{s_2}(\tau, t_{i}) \right \Vert_{\diamond}\Bigg).\label{eq:bound-first-term-int}
    \end{align}
    The first equality follows from Lemma~\ref{lem:time-dep-deriv-for-err-bnd}. The first inequality uses the triangle inequality for the diamond norm, and the second uses submultiplicativity. In the last equality, we used that $\mathcal{E}_{s_2}$ is a quantum channel for $s_2\in[0,1]$, and hence has diamond norm one.
    From Lemma~\ref{lem:time-dep-deriv-for-err-bnd} and the same norm estimates as above, we have, for all $v\leq t$ and $s\in[0,1]$,
    \begin{align}
         \left \Vert \frac{d}{ds}  \mathcal{E}_s(t, v) \right \Vert_{\diamond} & = \left \Vert \int_v^t d\tau'\, \mathcal{E}_s(t, \tau') \circ \mathcal{L}(\tau') \circ \mathcal{E}_s(\tau', v)  \right \Vert_{\diamond}\\
         & \leq \int_v^t d\tau'\, \left \Vert \mathcal{E}_s(t, \tau') \right \Vert_{\diamond}  \left \Vert\mathcal{L}(\tau')\right \Vert_{\diamond}  \left \Vert\mathcal{E}_s(\tau', v)  \right \Vert_{\diamond}\\
         & = \int_v^t d\tau'\, \left \Vert\mathcal{L}(\tau')\right \Vert_{\diamond}.
         \label{eq:E-s-deriv-bound}
    \end{align}
    Applying~\eqref{eq:E-s-deriv-bound} to the two derivative terms in~\eqref{eq:bound-first-term-int}, we obtain
    \begin{align}
       \left \Vert \frac{d^2}{ds_2^2} \mathcal{E}_{s_2}(t_{i+1}, t_{i}) \right \Vert_{\diamond} & \leq \int_{t_{i}}^{t_{i+1}} d\tau\ \left \Vert \mathcal{L}(\tau) \right \Vert_{\diamond} \Bigg(\int_{\tau}^{t_{i+1}} d\tau'\ \left \Vert\mathcal{L}(\tau')\right \Vert_{\diamond} + \int_{t_{i}}^{\tau} d\tau'\ \left \Vert\mathcal{L}(\tau')\right \Vert_{\diamond} \Bigg)\\
       & = \left (\int_{t_{i}}^{t_{i+1}} d\tau\ \left \Vert \mathcal{L}(\tau) \right \Vert_{\diamond}\right)^2\\
       & = \left(\left \Vert \mathcal{L} \right \Vert_{\diamond, 1}^{(i)}\right)^2\label{eq:bound-first-term-sec}.
    \end{align}

Next we bound the second term in~\eqref{eq:second-order}. We have
\begin{align}
    \left \Vert \frac{d^2}{ds_2^2} \mathcal{A}_{s_2}(t_{i+1}, t_{i}) \right \Vert_{\diamond} & =  \left \Vert \frac{d^2}{ds_2^2} \left (\int_{t_i}^{t_{i+1}} d\tau\, p^{(i)}(\tau)\ e^{s_2\frac{\mathcal{L}(\tau)}{p^{(i)}(\tau)}} \right) \right \Vert_{\diamond}\\
    & \leq \int_{t_i}^{t_{i+1}} d\tau\, p^{(i)}(\tau)\  \left \Vert \frac{d^2}{ds_2^2} e^{s_2\frac{\mathcal{L}(\tau)}{p^{(i)}(\tau)}} \right \Vert_{\diamond}\\
    & = \int_{t_i}^{t_{i+1}} d\tau\, p^{(i)}(\tau)\  \left \Vert  \frac{\mathcal{L}(\tau)}{p^{(i)}(\tau)} \circ \frac{\mathcal{L}(\tau)}{p^{(i)}(\tau)} \circ e^{s_2\frac{\mathcal{L}(\tau)}{p^{(i)}(\tau)}} \right \Vert_{\diamond}\\
    & \leq \int_{t_i}^{t_{i+1}} d\tau\,  \frac{\left \Vert  \mathcal{L}(\tau) \right \Vert_{\diamond}^2}{p^{(i)}(\tau)} \left \Vert e^{s_2\frac{\mathcal{L}(\tau)}{p^{(i)}(\tau)}} \right \Vert_{\diamond}\\
    & = \int_{t_i}^{t_{i+1}} d\tau\, \frac{\left \Vert \mathcal{L}(\tau) \right \Vert_{\diamond}^2}{p^{(i)}(\tau)} \\
    & = \left(\left \Vert \mathcal{L} \right \Vert_{\diamond, 1}^{(i)}\right)^2.\label{eq:bound-sec-term-sec}
\end{align}
The first inequality follows from the triangle inequality for the diamond norm and the positivity of the density $p^{(i)}$. The second inequality uses submultiplicativity of the diamond norm. In the following equality, we used that $e^{s_2\mathcal{L}(\tau)/p^{(i)}(\tau)}$ is a quantum channel for $s_2\in[0,1]$, and hence has diamond norm one. The last equality follows from the definition of $p^{(i)}$ in~\eqref{eq:prob-den-i}.

Substituting the bounds~\eqref{eq:bound-first-term-sec} and~\eqref{eq:bound-sec-term-sec} into~\eqref{eq:second-order}, we obtain
\begin{align}
    \left \Vert \mathcal{E}(t_{i+1}, t_{i}) - \mathcal{A}(t_{i+1}, t_{i}) \right \Vert_{\diamond} \leq \int_0^1 ds_1\, \int_{0}^{s_1} ds_2\, 2 \left(\left \Vert \mathcal{L} \right \Vert_{\diamond, 1}^{(i)}\right)^2 = \left(\left \Vert \mathcal{L} \right \Vert_{\diamond, 1}^{(i)}\right)^2.
\end{align}
Applying this local bound to~\eqref{eq:sum-split-error} gives
\begin{align}
    \left \Vert \mathcal{E}(t, 0) - \prod_{i=0}^{r-1} \mathcal{A}(t_{i+1}, t_{i}) \right \Vert_{\diamond} \leq \sum_{i=0}^{r-1} \left(\left \Vert \mathcal{L} \right \Vert_{\diamond, 1}^{(i)}\right)^2 = r \frac{\left \Vert \mathcal{L} \right \Vert_{\diamond, 1}^2}{r^2} = \frac{\left \Vert \mathcal{L} \right \Vert_{\diamond, 1}^2}{r},
\end{align}
where the first equality follows from the equipartition condition~\eqref{eq:partition}.
\end{proof}

\subsection{First-Order Sampling Algorithm for Linear Combination Input Model}\label{sec:first-order-lc}

In this section, we assume that the Lindbladian $\mathcal{L}(\tau)$, for all $\tau \in [0, t]$, has the following structure:
\begin{enumerate}
    \item The Hamiltonian $H(\tau)$ is represented as a sum of local Hamiltonians $H_j(\tau)$, that is,
    \begin{equation}
        H(\tau) = \sum_{j=1}^{J} H_j(\tau)\label{eq:local-hamil}.
    \end{equation}
    By ``local'', we mean that the operator is acting locally on a constant number of qubits.
    \item Each Lindblad operator $L_k(\tau)$ is a local operator.
\end{enumerate}
Defining the superoperators
\begin{align}
    \mathcal{H}_j(\tau)(\cdot) & \coloneqq -i [H_j(\tau), (\cdot)],\\
    \mathcal{L}_k(\tau)(\cdot) & \coloneqq L_k(\tau) (\cdot) L_k^\dagger(\tau) - \frac{1}{2} \left\{L_k^\dagger(\tau) L_k(\tau), (\cdot) \right\},
\end{align}
we can express the Lindbladian $\mathcal{L}(\tau)$ in the following form:
\begin{equation}
    \mathcal{L}(\tau) = \sum_{j=1}^J\underbrace{\mathcal{H}_j(\tau)}_{\eqqcolon \mathcal{M}_j(\tau)} + \sum_{k=1}^K \underbrace{\mathcal{L}_k(\tau)}_{\eqqcolon \mathcal{M}_{J+k}(\tau)} = \sum_{\ell = 1}^{J+K} \mathcal{M}_{\ell}(\tau)\label{eq:def-lind-lc}.
\end{equation}
We assume that each $\mathcal{M}_{\ell}(\tau)$ is bounded and piecewise continuous on $[0,t]$ and nonzero everywhere.

Our goal remains the simulation of the quantum channel $\mathcal{E}(t, 0) = \operatorname{exp}_{\mathcal{T}}\!\left(\int_0^t d\tau\, \mathcal{L}(\tau) \right)$ but now exploiting this additional structure of $\mathcal{L}(\tau)$. 
To measure the cost of the algorithm (gate complexity), we count the number of gates of the form $e^{c(\tau)\mathcal{M}_{\ell}(\tau)}$, where $c(\tau)$ is some coefficient that depends on the time $\tau$ and each gate $e^{c(\tau)\mathcal{M}_{\ell}(\tau)}$ either corresponds to a Hamiltonian simulation gate $e^{c(\tau)\mathcal{H}_{j}(\tau)}$ or a Lindbladian simulation gate $e^{c(\tau)\mathcal{L}_{k}(\tau)}$. 
We therefore treat these local evolutions as elementary operations in our cost model and count the number of such operations rather than tracking their decomposition into primitive one- and two-qubit gates.

This construction also extends the linear-combination continuous qDRIFT approach of~\cite{Berry2020TimeDependentHamiltonian} from time-dependent Hamiltonian simulation to local Hamiltonian and dissipative generators.
As in the preceding subsection, we define the following randomized quantum channel:
\begin{equation}
    \mathcal{A}^{\operatorname{(LC)}}(t, v) \coloneqq \sum_{\ell=1}^{J+K}\int_{v}^{t} d\tau\ p(\ell,\tau)\, e^{\frac{\mathcal{M}_{\ell}(\tau)}{p(\ell,\tau)}}\label{eq:rand-qc-lc},
\end{equation}
where $p(\ell,\tau)$ is a joint probability distribution over $\ell \in \{1,\ldots,J+K\}$ and $\tau \in [v, t]$, given by
\begin{equation}
     p(\ell,\tau)  \coloneqq 
        \frac{\left \Vert \mathcal{M}_{\ell}(\tau) \right \Vert_{\diamond}}{\left \Vert \mathcal{L} \right \Vert_{\diamond, 1, 1}},\label{eq:prob-den-lc}
\end{equation}
with 
\begin{equation}
    \left \Vert \mathcal{L} \right \Vert_{\diamond, 1, 1} \coloneqq \sum_{\ell=1}^{J+K} \int_{v}^{t} d\tau\, \left \Vert \mathcal{M}_{\ell}(\tau) \right \Vert_{\diamond}.
\end{equation}

The complete sampler is formed by composing such channels $r$ times:
\begin{equation}
    \prod_{i=0}^{r-1} \mathcal{A}^{\operatorname{(LC)}}(t_{i+1}, t_{i}) =   \prod_{i=0}^{r-1} \left[\sum_{\ell=1}^{J+K}\int_{t_{i}}^{t_{i+1}} d\tau\ p^{(i)}(\ell,\tau)\, e^{\frac{\mathcal{M}_{\ell}(\tau)}{p^{(i)}(\ell,\tau)}}\right]\label{eq:comp-protocol-lc},
\end{equation}
where each $p^{(i)}(\ell,\tau)$ is a joint probability distribution over $\ell\in\{1,\ldots,J+K\}$ and $\tau\in[t_i,t_{i+1}]$, defined analogously to~\eqref{eq:prob-den-lc} by
\begin{equation}
      p^{(i)}(\ell,\tau)  \coloneqq  
        \frac{\left \Vert \mathcal{M}_{\ell}(\tau) \right \Vert_{\diamond}}{\left \Vert \mathcal{L} \right \Vert_{\diamond, 1, 1}^{(i)}},\label{eq:prob-den-lc-i}
\end{equation}
with 
\begin{equation}
    \left \Vert \mathcal{L} \right \Vert_{\diamond, 1, 1}^{(i)} \coloneqq  \sum_{\ell=1}^{J+K} \int_{t_i}^{t_{i+1}} d\tau\ \left \Vert \mathcal{M}_{\ell}(\tau) \right \Vert_{\diamond}.
\end{equation}
We assume that, for each $i$, we can efficiently sample a pair $(\ell,\tau)$ from the joint distribution $p^{(i)}(\ell,\tau)$.
We choose the time points $0=t_0<t_1<\cdots<t_r=t$ such that, for all $i$,
\begin{equation}
    \left \Vert \mathcal{L} \right \Vert_{\diamond, 1, 1}^{(i)} = \frac{1}{r}\left \Vert \mathcal{L} \right \Vert_{\diamond, 1, 1}.\label{eq:partition-LC}
\end{equation}
Such a partition exists by applying the intermediate value theorem to the cumulative function $F(x) = \sum_{\ell=1}^{J+K} \int_0^x d\tau\, \left \Vert \mathcal{M}_{\ell}(\tau) \right \Vert_{\diamond}.$

The following theorem then shows how well the sampler $\prod_{i=0}^{r-1} \mathcal{A}^{\operatorname{(LC)}}(t_{i+1}, t_{i})$ approximates the target channel $\mathcal{E}(t, 0)$.

\begin{theorem} 
\label{thm:error-analysis-lc}
Let $\mathcal{L}(\tau)$ be a time-dependent Lindbladian as defined in~\eqref{eq:def-lind-lc} for all $\tau \in [0, t]$. We assume that each $\mathcal{M}_{\ell}(\tau)$ is bounded and piecewise continuous on $[0,t]$ and is nonzero everywhere. Let $\mathcal{E}(t, 0) =  \operatorname{exp}_{\mathcal{T}}\!\left(\int_0^t d\tau\, \mathcal{L}(\tau) \right)$ be the solution to the corresponding Lindblad master equation, and let $\mathcal{A}^{\operatorname{(LC)}}(t, v)$ be the randomized quantum channel as defined in~\eqref{eq:rand-qc-lc}. Let $0=t_0<t_1<\cdots<t_r=t$ be a partition satisfying~\eqref{eq:partition-LC}. Then, for every positive integer $r$, the following holds:    
\begin{equation}
        \left \Vert \mathcal{E}(t, 0) - \prod_{i=0}^{r-1} \mathcal{A}^{\operatorname{(LC)}}(t_{i+1}, t_{i}) \right \Vert_{\diamond} \leq \frac{\left \Vert \mathcal{L} \right \Vert_{\diamond, 1, 1}^2}{r}. 
\end{equation}
\end{theorem}
\begin{proof}
The proof follows a similar line of reasoning as that of Theorem~\ref{thm:error-analysis-gen}, but we include it here for completeness.

Using the composition property of the time-ordered propagator, together with a telescoping expansion and the fact that quantum channels have diamond norm one, we obtain
    \begin{align}
        \left \Vert \mathcal{E}(t, 0) - \prod_{i=0}^{r-1} \mathcal{A}^{\operatorname{(LC)}}(t_{i+1}, t_{i}) \right \Vert_{\diamond} \leq \sum_{i=0}^{r-1} \left \Vert \mathcal{E}(t_{i+1}, t_{i}) - \mathcal{A}^{\operatorname{(LC)}}(t_{i+1}, t_{i}) \right \Vert_{\diamond}\label{eq:sum-split-error-lc}.
    \end{align}
For each $i\in \{0, \ldots , r-1\}$ and $s\in[0,1]$, we define a scaled version of the quantum channel $\mathcal{A}^{\operatorname{(LC)}}(t_{i+1}, t_{i})$:
\begin{equation}
    \mathcal{A}^{\operatorname{(LC)}}_s(t_{i+1}, t_{i}) \coloneqq \sum_{\ell = 1}^{J+K}\int_{t_i}^{t_{i+1}} d\tau\, p^{(i)}(\ell,\tau)\ e^{s\frac{\mathcal{M}_{\ell}(\tau)}{p^{(i)}(\ell,\tau)}}.\label{eq:scaled-rand-qc-lc}
\end{equation}
From this, we have  $\mathcal{A}^{\operatorname{(LC)}}_1(t_{i+1}, t_{i}) =  \mathcal{A}^{\operatorname{(LC)}}(t_{i+1}, t_{i})$ and $\mathcal{A}^{\operatorname{(LC)}}_0(t_{i+1}, t_{i}) = \mathcal{I}$. Similarly, using the scaled version of the channel $\mathcal{E}(t, v)$, as defined in~\eqref{eq:scaled-E-channel}, for all $i \in \{0,\ldots,r-1\}$, we have $ \mathcal{E}_1(t_{i+1}, t_{i}) =  \mathcal{E}(t_{i+1}, t_{i})$ and $\mathcal{E}_0(t_{i+1}, t_{i}) = \mathcal{I}$.

Now using the fundamental theorem of calculus and following a similar line of reasoning as that of~\eqref{eq:use-1}--\eqref{eq:first-fun-cal}, we get
\begin{align}
         \left \Vert \mathcal{E}(t_{i+1}, t_{i}) - \mathcal{A}^{\operatorname{(LC)}}(t_{i+1}, t_{i}) \right \Vert_{\diamond} &  =  \left \Vert \int_0^1 ds_1\, \frac{d}{ds_1} \left(\mathcal{E}_{s_1}(t_{i+1}, t_{i}) - \mathcal{A}^{\operatorname{(LC)}}_{s_1}(t_{i+1}, t_{i}) \right) \right \Vert_{\diamond}\label{eq:first-fun-cal-lc}.
\end{align}
Also, note that the derivatives $\frac{d}{ds_1} \mathcal{E}_{s_1}(t_{i+1}, t_{i}) $ and $\frac{d}{ds_1} \mathcal{A}^{\operatorname{(LC)}}_{s_1}(t_{i+1}, t_{i}) $ match at $s_1 = 0$:
\begin{align}
         \left.\frac{d}{ds_1} \mathcal{A}^{\operatorname{(LC)}}_{s_1}(t_{i+1}, t_{i})\right |_{s_1 = 0} & = \left.\frac{d}{ds_1} \sum_{\ell = 1}^{J+K}\int_{t_i}^{t_{i+1}} d\tau\, p^{(i)}(\ell,\tau)\, e^{s_1\frac{\mathcal{M}_{\ell}(\tau)}{p^{(i)}(\ell,\tau)}}\right |_{s_1 = 0}\\
         & = \sum_{\ell = 1}^{J+K} \int_{t_i}^{t_{i+1}} d\tau\, p^{(i)}(\ell,\tau)\ \frac{\mathcal{M}_{\ell}(\tau)}{p^{(i)}(\ell,\tau)}\ \left.e^{s_1\frac{\mathcal{M}_{\ell}(\tau)}{p^{(i)}(\ell,\tau)}}\right |_{s_1 = 0}\\
         & = \int_{t_i}^{t_{i+1}} d\tau\, \sum_{\ell = 1}^{J+K} \mathcal{M}_{\ell}(\tau)\ \left. e^{s_1\frac{\mathcal{M}_{\ell}(\tau)}{p^{(i)}(\ell,\tau)}}\right |_{s_1 = 0}\\
         & = \int_{t_{i}}^{t_{i+1}} d\tau\, \mathcal{L}(\tau)\\
         & = \left.\frac{d}{ds_1} \mathcal{E}_{s_1}(t_{i+1}, t_{i})\right |_{s_1 = 0},
    \end{align}
where the last equality follows from~\eqref{eq:first-order-E-identity}.
Therefore, using the fundamental theorem of calculus again in~\eqref{eq:first-fun-cal-lc}, we get
 \begin{align}
        \left \Vert \mathcal{E}(t_{i+1}, t_{i}) - \mathcal{A}^{\operatorname{(LC)}}(t_{i+1}, t_{i}) \right \Vert_{\diamond}
        & = \left \Vert \int_0^1 ds_1\, \int_{0}^{s_1} ds_2\, \frac{d^2}{ds_2^2} \left(\mathcal{E}_{s_2}(t_{i+1}, t_{i}) - \mathcal{A}^{\operatorname{(LC)}}_{s_2}(t_{i+1}, t_{i}) \right) \right \Vert_{\diamond}\\
        & \leq  \int_0^1 ds_1\,  \int_{0}^{s_1}  ds_2\,\left( \left \Vert \frac{d^2}{ds_2^2} \mathcal{E}_{s_2}(t_{i+1}, t_{i}) \right \Vert_{\diamond}  + \left \Vert \frac{d^2}{ds_2^2} \mathcal{A}^{\operatorname{(LC)}}_{s_2}(t_{i+1}, t_{i}) \right \Vert_{\diamond}\right).\label{eq:second-order-lc}
    \end{align}

Next, we bound these two terms sequentially from above. We already have a bound on the first term from the proof of Theorem~\ref{thm:error-analysis-gen} (see~\eqref{eq:bound-first-term-sec}):
\begin{equation}
    \left \Vert \frac{d^2}{ds_2^2} \mathcal{E}_{s_2}(t_{i+1}, t_{i}) \right \Vert_{\diamond} \leq \left(\left \Vert \mathcal{L} \right \Vert_{\diamond, 1}^{(i)}\right)^2 \leq \left(\left \Vert \mathcal{L} \right \Vert_{\diamond, 1, 1}^{(i)}\right)^2 \label{eq:first-term-bound-lc}.
\end{equation}
For the second term, we follow the same argument as in~\eqref{eq:bound-sec-term-sec}, now applied term by term to the linear-combination sampler. 
Since each $\mathcal{M}_{\ell}(\tau)$ is either a Hamiltonian generator or a Lindbladian generator, the map $e^{s_2\mathcal{M}_{\ell}(\tau)/p^{(i)}(\ell,\tau)}$ is a quantum channel for $s_2\in[0,1]$, and hence has diamond norm one. 
Thus,
\begin{align}
    \left \Vert \frac{d^2}{ds_2^2} \mathcal{A}^{\operatorname{(LC)}}_{s_2}(t_{i+1}, t_{i}) \right \Vert_{\diamond} & =  \left \Vert \frac{d^2}{ds_2^2} \left (\sum_{\ell = 1}^{J+K} \int_{t_i}^{t_{i+1}} d\tau\, p^{(i)}(\ell,\tau)\ e^{s_2\frac{\mathcal{M}_{\ell}(\tau)}{p^{(i)}(\ell,\tau)}} \right) \right \Vert_{\diamond}\\
    & \leq \sum_{\ell = 1}^{J+K} \int_{t_i}^{t_{i+1}} d\tau\, p^{(i)}(\ell,\tau)\, \left \Vert \frac{d^2}{ds_2^2} e^{s_2\frac{\mathcal{M}_{\ell}(\tau)}{p^{(i)}(\ell,\tau)}} \right \Vert_{\diamond}\\
    & = \sum_{\ell = 1}^{J+K} \int_{t_i}^{t_{i+1}} d\tau\, p^{(i)}(\ell,\tau)\, \left \Vert  \frac{\mathcal{M}_{\ell}(\tau)}{p^{(i)}(\ell,\tau)} \circ \frac{\mathcal{M}_{\ell}(\tau)}{p^{(i)}(\ell,\tau)} \circ e^{s_2\frac{\mathcal{M}_{\ell}(\tau)}{p^{(i)}(\ell,\tau)}} \right \Vert_{\diamond}\\
    & \leq \sum_{\ell = 1}^{J+K} \int_{t_i}^{t_{i+1}} d\tau\, \frac{\left \Vert \mathcal{M}_{\ell}(\tau) \right \Vert_{\diamond}^2}{p^{(i)}(\ell,\tau)} \left \Vert e^{s_2\frac{\mathcal{M}_{\ell}(\tau)}{p^{(i)}(\ell,\tau)}} \right \Vert_{\diamond}\\
    & = \sum_{\ell = 1}^{J+K} \int_{t_i}^{t_{i+1}} d\tau\,  \frac{\left \Vert \mathcal{M}_{\ell}(\tau) \right \Vert_{\diamond}^2}{p^{(i)}(\ell,\tau)} \\
    & = \left(\left \Vert \mathcal{L} \right \Vert_{\diamond, 1, 1}^{(i)}\right)^2 \label{eq:sec-term-bound-lc}.
\end{align}

Substituting the bounds~\eqref{eq:first-term-bound-lc} and~\eqref{eq:sec-term-bound-lc} into~\eqref{eq:second-order-lc}, we obtain
\begin{equation}
    \left \Vert \mathcal{E}(t_{i+1}, t_{i}) - \mathcal{A}^{\operatorname{(LC)}}(t_{i+1}, t_{i}) \right \Vert_{\diamond}\leq  \int_0^1 ds_1\,  \int_{0}^{s_1}  ds_2\, 2\left(\left \Vert \mathcal{L} \right \Vert_{\diamond, 1, 1}^{(i)}\right)^2=\left(\left \Vert \mathcal{L} \right \Vert_{\diamond, 1, 1}^{(i)}\right)^2.
\end{equation}
Applying this local bound to~\eqref{eq:sum-split-error-lc} gives
\begin{align}
    \left \Vert \mathcal{E}(t, 0) - \prod_{i=0}^{r-1} \mathcal{A}^{\operatorname{(LC)}}(t_{i+1}, t_{i}) \right \Vert_{\diamond} \leq \sum_{i=0}^{r-1} \left(\left \Vert \mathcal{L} \right \Vert_{\diamond, 1, 1}^{(i)}\right)^2 =   \frac{\left \Vert \mathcal{L} \right \Vert_{\diamond, 1, 1}^{2}}{r},
\end{align}
where the equality follows from the equipartition condition~\eqref{eq:partition-LC}.
\end{proof}

Recall that we aim to find the gate complexity, that is, the number of gates of the form $e^{c(\tau)\mathcal{M}_{\ell}(\tau)}$. The following result regarding the gate complexity immediately follows from the above theorem:
\begin{corollary}
    Let $\mathcal{L}(\tau) = \sum_{\ell=1}^{J+K} \mathcal{M}_{\ell}(\tau)$ be a time-dependent Lindbladian as defined in~\eqref{eq:def-lind-lc}, for some $J, K \in \mathbb{N}$ and for all $\tau \in [0, t]$. Each summand $\mathcal{M}_{\ell}(\tau)$ is either a Hamiltonian generator or a Lindbladian generator, and we assume that it is bounded and piecewise continuous on $[0,t]$ and nonzero everywhere. Then it suffices to use
    \begin{equation}
        \left\lceil
        \frac{\left \Vert \mathcal{L} \right \Vert_{\diamond,1,1}^{2}}
        {\varepsilon}
        \right\rceil
    \end{equation}
    gates of the form $e^{c(\tau)\mathcal{M}_{\ell}(\tau)}$ to simulate $\mathcal{E}(t,0)$ with diamond norm error at most $\varepsilon$, provided that one can efficiently sample from the joint distributions defined in~\eqref{eq:prob-den-lc-i}.
\end{corollary}

\subsubsection{Special Case}
We now focus on a physically motivated special case of the general linear-combination input model considered above. In this scenario, both the time-dependent Hamiltonian and Lindblad operators factorize into time-independent components modulated by scalar coefficient functions:
\begin{align}
    H(\tau) & = \sum_{j=1}^{J} \alpha_j(\tau) H_j\\
    L_k(\tau) & = \beta_k(\tau) L_k,
\end{align}
where $H_j$ and $L_k$ are time-independent operators and the coefficients are smooth functions $\alpha_j(\tau)$ and $\beta_k(\tau)$ that completely capture the time dependence. This structure emerges naturally in many physical systems where the types of interactions stay the same, but their strengths vary with time. We measure the cost of an algorithm in terms of the number of gates of the form $e^{c(\tau)\mathcal{M}}$, where $\mathcal{M}$ is now a time-independent Hamiltonian or Lindbladian generator.

To simulate this special model, we essentially employ the same sampler used for the general linear-combination input model, but now the sampling distribution can be written more explicitly in terms of the scalar coefficient functions.

First, define the time-independent superoperators
\begin{align}
    \mathcal{H}_j(\cdot) &\coloneqq -i[H_j,(\cdot)],\\
    \mathcal{L}_k(\cdot) &\coloneqq L_k(\cdot)L_k^\dagger -\frac{1}{2}\{L_k^\dagger L_k,(\cdot)\}.
\end{align}
Then the corresponding time-dependent Lindbladian can be written as
\begin{equation}
    \mathcal{L}(\tau) = \sum_{j=1}^{J} \alpha_j(\tau)\mathcal{H}_j +\sum_{k=1}^{K}|\beta_k(\tau)|^2\mathcal{L}_k,
\end{equation}
where $H_j$ and $L_k$ are time-independent operators, $\alpha_j(\tau)$ are time-dependent real-valued functions, and $\beta_k(\tau)$ are time-dependent scalar functions.
Equivalently, we can write
\begin{equation}
    \mathcal{L}(\tau) = \sum_{\ell=1}^{J+K} \gamma_\ell(\tau)\mathcal{M}_\ell,
\end{equation}
where
\begin{equation}
    \gamma_{\ell}(\tau) = 
    \begin{cases}
          \alpha_{\ell}(\tau), & 1\leq \ell\leq J,\\
          \left|\beta_{\ell-J}(\tau)\right|^2, & J+1\leq \ell \leq J+K,
    \end{cases}
\end{equation}
and
\begin{equation}
    \mathcal{M}_{\ell} = 
    \begin{cases}
          \mathcal{H}_{\ell}, & 1\leq \ell\leq J,\\
          \mathcal{L}_{\ell-J}, & J+1\leq \ell \leq J+K.
    \end{cases}
\end{equation}
For an interval $[v,t]$, define
\begin{equation}
    \left\Vert \mathcal{L} \right\Vert_{\diamond,1,1}^{\operatorname{(SC)}} \coloneqq \sum_{\ell=1}^{J+K} \int_v^t d\tau\, |\gamma_\ell(\tau)| \left\Vert \mathcal{M}_\ell \right\Vert_{\diamond}.
\end{equation}
The one-step averaged sampler for this special case is
\begin{equation}
    \mathcal{A}^{\operatorname{(SC)}}(t,v) \coloneqq \sum_{\ell=1}^{J+K} \int_v^t d\tau\, p(\ell,\tau) e^{\frac{\gamma_\ell(\tau)\mathcal{M}_\ell}{p(\ell,\tau)}},
\end{equation}
where
\begin{equation}
    p(\ell,\tau) \coloneqq \frac{|\gamma_\ell(\tau)|\left\Vert \mathcal{M}_\ell\right\Vert_{\diamond}}{\left\Vert \mathcal{L}\right\Vert_{\diamond,1,1}^{\operatorname{(SC)}}}.\label{eq:prob-den-SC}
\end{equation}

The complete sampler is formed by composing such channels $r$ times:
\begin{equation}
    \prod_{i=0}^{r-1} \mathcal{A}^{\operatorname{(SC)}}(t_{i+1}, t_{i}) = \prod_{i=0}^{r-1} \left[\sum_{\ell=1}^{J+K}\int_{t_{i}}^{t_{i+1}} d\tau\, p^{(i)}(\ell,\tau)\, e^{\frac{\gamma_{\ell}(\tau)\mathcal{M}_{\ell}}{p^{(i)}(\ell,\tau)}}\right]\label{eq:comp-protocol-SC},
\end{equation}
where each $p^{(i)}(\ell,\tau)$ is a joint probability distribution over $\ell\in\{1,\ldots,J+K\}$ and $\tau\in[t_i,t_{i+1}]$, defined analogously to~\eqref{eq:prob-den-SC} by
\begin{equation}
      p^{(i)}(\ell,\tau)  \coloneqq  
        \frac{\left|\gamma_{\ell}(\tau)\right|\left \Vert \mathcal{M}_{\ell} \right \Vert_{\diamond}}{\left \Vert \mathcal{L} \right \Vert_{\diamond, 1, 1}^{\operatorname{(SC)},(i)}},\label{eq:prob-den-SC-i}
\end{equation}
with 
\begin{equation}
    \left \Vert \mathcal{L} \right \Vert_{\diamond, 1, 1}^{\operatorname{(SC)},(i)} \coloneqq  \sum_{\ell=1}^{J+K} \int_{t_i}^{t_{i+1}} d\tau\,\left|\gamma_{\ell}(\tau)\right| \left \Vert \mathcal{M}_{\ell} \right \Vert_{\diamond}.
\end{equation}
We assume that, for each $i$, we can efficiently sample a pair $(\ell,\tau)$ from the joint distribution $p^{(i)}(\ell,\tau)$.
We choose the time points $0=t_0<t_1<\cdots<t_r=t$ such that, for all $i$,
\begin{equation}
    \left \Vert \mathcal{L} \right \Vert_{\diamond, 1, 1}^{\operatorname{(SC)},(i)} = \frac{1}{r}\left \Vert \mathcal{L} \right \Vert_{\diamond, 1, 1}^{\operatorname{(SC)}}.\label{eq:partition-SC}
\end{equation}
Such a partition exists by applying the intermediate value theorem to the cumulative function $F(x) = \sum_{\ell=1}^{J+K} \int_0^x d\tau\, \left|\gamma_{\ell}(\tau)\right|\left \Vert \mathcal{M}_{\ell} \right \Vert_{\diamond}.$

\begin{corollary}
    Consider the special case above, where $\mathcal{L}(\tau) = \sum_{\ell=1}^{J+K} \gamma_\ell(\tau)\mathcal{M}_\ell$.
    Suppose that the time-dependent summands $\gamma_{\ell}(\tau)\mathcal{M}_{\ell}$ satisfy the assumptions of Theorem~\ref{thm:error-analysis-lc}.
    Then, for every positive integer $r$, the sampler $\prod_{i=0}^{r-1}\mathcal{A}^{\operatorname{(SC)}}(t_{i+1},t_i)$ satisfies
    \begin{equation}
        \left\Vert \mathcal{E}(t,0) - \prod_{i=0}^{r-1}\mathcal{A}^{\operatorname{(SC)}}(t_{i+1},t_i)\right\Vert_{\diamond}
        \leq \frac{\left(\left \Vert \mathcal{L} \right \Vert_{\diamond, 1, 1}^{\operatorname{(SC)}}\right)^2}{r}.
    \end{equation}
    Consequently, it suffices to use
    \begin{equation}
        \left\lceil\frac{\left(\left \Vert \mathcal{L} \right \Vert_{\diamond, 1, 1}^{\operatorname{(SC)}}\right)^2}{\varepsilon}\right\rceil
    \end{equation}
    gates of the form $e^{c(\tau)\mathcal{M}_\ell}$ to simulate $\mathcal{E}(t,0)$ with diamond norm error at most $\varepsilon$, provided that one can efficiently sample from the joint distributions defined in~\eqref{eq:prob-den-SC-i}.
\end{corollary}

\subsection{Second-Order Sampling Algorithm}\label{sec:sec-order-sampling-algorithm}

From Section~\ref{sec:first-order-sampling-algorithm}, recall the first-order sampling algorithm given by~\eqref{eq:comp-protocol}:
\begin{equation}
    \prod_{i=0}^{r-1} \mathcal{A}(t_{i+1}, t_{i}) =  \prod_{i=0}^{r-1} \left[\int_{t_i}^{t_{i+1}} d\tau\ p^{(i)}(\tau)\, e^{\frac{\mathcal{L}(\tau)}{p^{(i)}(\tau)}}\right].
\end{equation}
Also, recall that we choose the time points $0=t_0  < t_1 < \cdots < t_{r-1} < t_r = t $ such that the equality in~\eqref{eq:partition} holds for all $i$.
Next, we define the following map:
\begin{equation}
    \mathcal{R}(t_{i+1}, t_{i}) \coloneqq  \frac{1}{2}\left(\left.\frac{d^2}{ds^2} \mathcal{E}_{s}(t_{i+1}, t_{i})\right |_{s=0}  - \left.\frac{d^2}{ds^2} \mathcal{A}_{s}(t_{i+1}, t_{i})\right |_{s=0}\right).\label{eq:second-order-r-map}
\end{equation}
Using this correction, we define the map corresponding to our second-order sampling algorithm as
\begin{equation}
    \mathcal{A}^{(2)}(t,0)
    \coloneqq
    \prod_{i=0}^{r-1}\mathcal{A}(t_{i+1},t_i)
    + \sum_{k=0}^{r-1}
    \left(\prod_{i=k+1}^{r-1}\mathcal{A}(t_{i+1},t_i)\right)
    \mathcal{R}(t_{k+1},t_k)
    \left(\prod_{i=0}^{k-1}\mathcal{A}(t_{i+1},t_i)\right).
    \label{eq:sec-order-sampling-map}
\end{equation}
Here, an empty product is understood as the identity map, and the products are ordered so that the map corresponding to the earliest interval acts first.

Next, in the following theorem, we analyze how well this map approximates the target channel $\mathcal{E}(t,0)$.
\begin{theorem}[Second-order sampling error]\label{thm:sec-order-sampler}
    Let $\mathcal{L}(\tau)$ be a time-dependent Lindbladian on $\tau\in[0,t]$ satisfying the assumptions of Theorem~\ref{thm:error-analysis-gen}. Let $0=t_0<t_1<\cdots<t_r=t$ be a partition satisfying~\eqref{eq:partition}, and let $\mathcal{A}^{(2)}(t,0)$ be the second-order sampling map defined in~\eqref{eq:sec-order-sampling-map}. Then,
    \begin{equation}
        \frac{1}{2}\left\Vert\mathcal{E}(t,0)-\mathcal{A}^{(2)}(t,0)\right\Vert_{\diamond}
        \leq \frac{\left\Vert\mathcal{L}\right\Vert_{\diamond,1}^3}{6r^2} + \frac{(r-1)\left\Vert\mathcal{L}\right\Vert_{\diamond,1}^4}{4r^3}.
    \end{equation}
    Consequently,
    \begin{equation}
        \frac{1}{2}\left\Vert\mathcal{E}(t,0)-\mathcal{A}^{(2)}(t,0)\right\Vert_{\diamond}
        = O\!\left(\frac{\Vert\mathcal{L}\Vert_{\diamond,1}^3+\Vert\mathcal{L}\Vert_{\diamond,1}^4}{r^2}\right).
    \end{equation}
\end{theorem}
\begin{proof}
We first bound the local error on each time interval after including the second-order correction. 
Using $\mathcal{E}_0(t_{i+1},t_i)=\mathcal{A}_0(t_{i+1},t_i)=\mathcal{I}$, the first-order matching identities~\eqref{eq:first-order-E-identity} and~\eqref{eq:first-order-A-identity}, and the definition of $\mathcal{R}(t_{i+1},t_i)$ in~\eqref{eq:second-order-r-map}, Taylor's formula with integral remainder gives
\begin{equation}
    \mathcal{E}(t_{i+1},t_i)-\mathcal{A}(t_{i+1},t_i)-\mathcal{R}(t_{i+1},t_i) =
    \int_0^1 ds\,\frac{(1-s)^2}{2}\frac{d^3}{ds^3}\left(\mathcal{E}_s(t_{i+1},t_i)-\mathcal{A}_s(t_{i+1},t_i)\right).
    \label{eq:td-second-order-local-remainder}
\end{equation}

To bound the third-order derivative of the exact propagator, we differentiate the propagator identity in Lemma~\ref{lem:time-dep-deriv-for-err-bnd} twice more. This gives
\begin{align}
    \frac{d^3}{ds^3}\mathcal{E}_s(t_{i+1},t_i)
    &=
    \int_{t_i}^{t_{i+1}}d\tau\,\Bigg[\frac{d^2}{ds^2}\mathcal{E}_s(t_{i+1},\tau)\circ\mathcal{L}(\tau)\circ\mathcal{E}_s(\tau,t_i)+2\frac{d}{ds}\mathcal{E}_s(t_{i+1},\tau)\circ\mathcal{L}(\tau)\circ\frac{d}{ds}\mathcal{E}_s(\tau,t_i)\notag\\
    &\qquad
    +\mathcal{E}_s(t_{i+1},\tau)\circ\mathcal{L}(\tau)\circ\frac{d^2}{ds^2}\mathcal{E}_s(\tau,t_i)\Bigg].
\end{align}
Applying the first- and second-derivative norm estimates in~\eqref{eq:bound-first-term}--\eqref{eq:bound-first-term-sec}, together with submultiplicativity of the diamond norm, we obtain
\begin{align}
    \left\Vert\frac{d^3}{ds^3}\mathcal{E}_s(t_{i+1},t_i)\right\Vert_{\diamond}
    &\leq
    \int_{t_i}^{t_{i+1}}d\tau\,\left\Vert\mathcal{L}(\tau)\right\Vert_{\diamond}\Bigg[\left(\int_{\tau}^{t_{i+1}}d\tau'\,\left\Vert\mathcal{L}(\tau')\right\Vert_{\diamond}\right)^2\notag\\
    &\qquad
    +2\left(\int_{\tau}^{t_{i+1}}d\tau'\,\left\Vert\mathcal{L}(\tau')\right\Vert_{\diamond}\right)\left(\int_{t_i}^{\tau}d\tau'\,\left\Vert\mathcal{L}(\tau')\right\Vert_{\diamond}\right)\notag\\
    &\qquad
    +\left(\int_{t_i}^{\tau}d\tau'\,\left\Vert\mathcal{L}(\tau')\right\Vert_{\diamond}\right)^2\Bigg]\\
    &=
    \int_{t_i}^{t_{i+1}}d\tau\,\left\Vert\mathcal{L}(\tau)\right\Vert_{\diamond}\left(\int_{t_i}^{t_{i+1}}d\tau'\,\left\Vert\mathcal{L}(\tau')\right\Vert_{\diamond}\right)^2\\
    &=
    \left(\left\Vert\mathcal{L}\right\Vert_{\diamond,1}^{(i)}\right)^3.
    \label{eq:sec-o-proof-first-term-decomp-1}
\end{align}
Similarly, differentiating the randomized channel under the integral sign, as in the derivation of~\eqref{eq:bound-sec-term-sec}, and using submultiplicativity of the diamond norm gives
\begin{align}
    \left\Vert\frac{d^3}{ds^3}\mathcal{A}_s(t_{i+1},t_i)\right\Vert_{\diamond}
    &\leq
    \int_{t_i}^{t_{i+1}}d\tau\,\frac{\left\Vert\mathcal{L}(\tau)\right\Vert_{\diamond}^3}{\left(p^{(i)}(\tau)\right)^2}\\
    &=
    \left(\left\Vert\mathcal{L}\right\Vert_{\diamond,1}^{(i)}\right)^3.
    \label{eq:sec-o-proof-first-term-decomp-2}
\end{align}
Applying the bounds in~\eqref{eq:sec-o-proof-first-term-decomp-1} and~\eqref{eq:sec-o-proof-first-term-decomp-2} to the integral remainder in~\eqref{eq:td-second-order-local-remainder}, and using $\int_0^1(1-s)^2/2\,ds=1/6$, we obtain
\begin{equation}
    \left\Vert\mathcal{E}(t_{i+1},t_i)-\mathcal{A}(t_{i+1},t_i)-\mathcal{R}(t_{i+1},t_i)\right\Vert_{\diamond}
    \leq
    \frac{1}{3}\left(\left\Vert\mathcal{L}\right\Vert_{\diamond,1}^{(i)}\right)^3.
    \label{eq:td-second-order-local-residual}
\end{equation}

We also require a bound on the first-order local error. Since the exact and randomized maps agree through first order in $s$, Taylor's formula gives
\begin{equation}
    \mathcal{E}(t_{i+1},t_i)-\mathcal{A}(t_{i+1},t_i)
    = \int_0^1 ds\,(1-s)\frac{d^2}{ds^2}\left(\mathcal{E}_s(t_{i+1},t_i)-\mathcal{A}_s(t_{i+1},t_i)\right).
\end{equation}
Using the second-derivative bounds in~\eqref{eq:bound-first-term-sec} and~\eqref{eq:bound-sec-term-sec}, together with $\int_0^1(1-s)\,ds=1/2$, we obtain
\begin{equation}
    \left\Vert\mathcal{E}(t_{i+1},t_i)-\mathcal{A}(t_{i+1},t_i)\right\Vert_{\diamond}
    \leq
    \left(\left\Vert\mathcal{L}\right\Vert_{\diamond,1}^{(i)}\right)^2.
    \label{eq:td-first-order-local-bound}
\end{equation}
Moreover, from the definition of $\mathcal{R}(t_{i+1},t_i)$ in~\eqref{eq:second-order-r-map} and the same second-derivative bounds,
\begin{align}
    \left\Vert\mathcal{R}(t_{i+1},t_i)\right\Vert_{\diamond}
    &\leq
    \frac{1}{2}\left(\left\Vert\left.\frac{d^2}{ds^2}\mathcal{E}_s(t_{i+1},t_i)\right|_{s=0}\right\Vert_{\diamond} + \left\Vert\left.\frac{d^2}{ds^2}\mathcal{A}_s(t_{i+1},t_i)\right|_{s=0}\right\Vert_{\diamond}\right)\\
    &\leq
    \left(\left\Vert\mathcal{L}\right\Vert_{\diamond,1}^{(i)}\right)^2.
    \label{eq:R-bound}
\end{align}

We now compare the full exact evolution with the second-order sampling map. 
Using the composition property of the time-ordered propagator and the definition of $\mathcal{A}^{(2)}(t,0)$ in~\eqref{eq:sec-order-sampling-map}, a telescoping expansion of the difference between the exact and first-order products gives
\begin{align}
    \mathcal{E}(t,0)-\mathcal{A}^{(2)}(t,0)
    &=
    \sum_{k=0}^{r-1}
    \left(\prod_{i=k+1}^{r-1}\mathcal{E}(t_{i+1},t_i)\right)
    \left[
    \mathcal{E}(t_{k+1},t_k)
    -\mathcal{A}(t_{k+1},t_k)
    -\mathcal{R}(t_{k+1},t_k)
    \right]
    \left(\prod_{i=0}^{k-1}\mathcal{A}(t_{i+1},t_i)\right)
    \notag\\
    &\quad+
    \sum_{k=0}^{r-1}
    \left[
    \prod_{i=k+1}^{r-1}\mathcal{E}(t_{i+1},t_i)
    -
    \prod_{i=k+1}^{r-1}\mathcal{A}(t_{i+1},t_i)
    \right]
    \mathcal{R}(t_{k+1},t_k)
    \left(\prod_{i=0}^{k-1}\mathcal{A}(t_{i+1},t_i)\right).
    \label{eq:td-second-order-global-decomposition}
\end{align}
The first sum contains the residual local error after subtracting the second-order correction. The second sum accounts for the difference between the exact and first-order sampled evolutions following the inserted correction.

Since both $\mathcal{E}(t_{i+1},t_i)$ and $\mathcal{A}(t_{i+1},t_i)$ are quantum channels, all surrounding factors in the first sum have diamond norm one. 
Therefore, using~\eqref{eq:td-second-order-local-residual}, its contribution to the normalized diamond-norm error is bounded by
\begin{equation}
    \frac{1}{2}
    \sum_{k=0}^{r-1}
    \left\Vert
    \mathcal{E}(t_{k+1},t_k)
    -\mathcal{A}(t_{k+1},t_k)
    -\mathcal{R}(t_{k+1},t_k)
    \right\Vert_{\diamond}
    \leq
    \frac{1}{6}
    \sum_{k=0}^{r-1}
    \left(\left\Vert\mathcal{L}\right\Vert_{\diamond,1}^{(k)}\right)^3.
\end{equation}

For the second sum, we again use a telescoping expansion. Since the exact and randomized maps on each interval are quantum channels,~\eqref{eq:td-first-order-local-bound} gives
\begin{align}
    \left\Vert\prod_{i=k+1}^{r-1}\mathcal{E}(t_{i+1},t_i) - \prod_{i=k+1}^{r-1}\mathcal{A}(t_{i+1},t_i)\right\Vert_{\diamond}
    &\leq
    \sum_{j=k+1}^{r-1} \left\Vert\mathcal{E}(t_{j+1},t_j) - \mathcal{A}(t_{j+1},t_j)\right\Vert_{\diamond}\\
    &\leq
    \sum_{j=k+1}^{r-1}\left(\left\Vert\mathcal{L}\right\Vert_{\diamond,1}^{(j)}\right)^2.
\end{align}
Combining this bound with~\eqref{eq:R-bound} and the decomposition in~\eqref{eq:td-second-order-global-decomposition}, we obtain
\begin{equation}
    \frac{1}{2}\left\Vert\mathcal{E}(t,0)-\mathcal{A}^{(2)}(t,0)\right\Vert_{\diamond}
    \leq
    \frac{1}{6}\sum_{k=0}^{r-1}\left(\left\Vert\mathcal{L}\right\Vert_{\diamond,1}^{(k)}\right)^3 
    + \frac{1}{2}\sum_{k=0}^{r-1}\left[\sum_{j=k+1}^{r-1}\left(\left\Vert\mathcal{L}\right\Vert_{\diamond,1}^{(j)}\right)^2\right]\left(\left\Vert\mathcal{L}\right\Vert_{\diamond,1}^{(k)}\right)^2.
\end{equation}

Finally, by the choice of partition in~\eqref{eq:partition}, $\left\Vert\mathcal{L}\right\Vert_{\diamond,1}^{(i)} =\left\Vert\mathcal{L}\right\Vert_{\diamond,1}/r$ for every $i$. 
Therefore, the first term becomes
\begin{equation}
    \frac{1}{6}\sum_{k=0}^{r-1}\left(\left\Vert\mathcal{L}\right\Vert_{\diamond,1}^{(k)}\right)^3 = \frac{\left\Vert\mathcal{L}\right\Vert_{\diamond,1}^3}{6r^2},
\end{equation}
while the second term contains one contribution for each index pair $(k,j)$ satisfying $0\leq k<j\leq r-1$. There are $\binom{r}{2}=r(r-1)/2$ such pairs, and therefore
\begin{equation}
    \frac{1}{2}\sum_{k=0}^{r-1}\left[\sum_{j=k+1}^{r-1}\left(\left\Vert\mathcal{L}\right\Vert_{\diamond,1}^{(j)}\right)^2\right]\left(\left\Vert\mathcal{L}\right\Vert_{\diamond,1}^{(k)}\right)^2
    = \frac{1}{2}\binom{r}{2}\frac{\left\Vert\mathcal{L}\right\Vert_{\diamond,1}^4}{r^4}
    = \frac{(r-1)\left\Vert\mathcal{L}\right\Vert_{\diamond,1}^4}{4r^3}.
\end{equation}
Thus,
\begin{equation}
    \frac{1}{2}\left\Vert\mathcal{E}(t,0)-\mathcal{A}^{(2)}(t,0)\right\Vert_{\diamond}
    \leq\frac{\left\Vert\mathcal{L}\right\Vert_{\diamond,1}^3}{6r^2} 
    + \frac{(r-1)\left\Vert\mathcal{L}\right\Vert_{\diamond,1}^4}{4r^3}.
\end{equation}
Since $(r-1)/r\leq1$, it follows that
\begin{equation}
    \frac{1}{2}\left\Vert\mathcal{E}(t,0)-\mathcal{A}^{(2)}(t,0)\right\Vert_{\diamond}
    = O\!\left(\frac{\left\Vert\mathcal{L}\right\Vert_{\diamond,1}^3 + \left\Vert\mathcal{L}\right\Vert_{\diamond,1}^4}{r^2}\right).
\end{equation}
This concludes the proof.
\end{proof}

\subsection{Second-Order Sampling Algorithm for Linear Combination Input Model}

Consider the setting described in Section~\ref{sec:first-order-lc}, where the Lindbladian $\mathcal{L}(\tau)$ is expressed as a sum of local terms $\mathcal{M}_{\ell}(\tau)$ as in~\eqref{eq:def-lind-lc}. Each local term $\mathcal{M}_{\ell}(\tau)$ corresponds either to a local Hamiltonian or a local Lindbladian. Furthermore, as mentioned in Section~\ref{sec:first-order-lc}, we evaluate the cost of an algorithm by counting local operations of the form $e^{c(\tau)\mathcal{M}_{\ell}(\tau)}$, where $c(\tau)$ is a time-dependent coefficient. Also recall that the first-order sampler in the linear combination input model is defined by
\begin{equation}
    \prod_{i=0}^{r-1} \mathcal{A}^{\operatorname{(LC)}}(t_{i+1}, t_i),
\end{equation}
where the channel $\mathcal{A}^{\operatorname{(LC)}}(t, v)$ is given in~\eqref{eq:rand-qc-lc}.

Next, we define the following map that we employ as the second-order correction:
\begin{equation}
    \mathcal{R}^{\operatorname{(LC)}}(t_{i+1}, t_i) \coloneqq \frac{1}{2}\left(\left. \frac{d^2}{ds^2} \mathcal{E}_{s}(t_{i+1}, t_i) \right|_{s=0} - \left. \frac{d^2}{ds^2} \mathcal{A}^{\operatorname{(LC)}}_{s}(t_{i+1}, t_i) \right|_{s=0}\right).
    \label{eq:second-order-r-map-lc}
\end{equation}
Using this correction, we define the second-order sampling map for the linear combination input model as
\begin{align}
    \mathcal{A}^{(2,\operatorname{LC})}(t,0)
    \coloneqq
    \prod_{i=0}^{r-1}\mathcal{A}^{\operatorname{(LC)}}(t_{i+1},t_i) + \sum_{k=0}^{r-1}\left(\prod_{i=k+1}^{r-1}\mathcal{A}^{\operatorname{(LC)}}(t_{i+1},t_i)\right)\mathcal{R}^{\operatorname{(LC)}}(t_{k+1},t_k)\left(\prod_{i=0}^{k-1}\mathcal{A}^{\operatorname{(LC)}}(t_{i+1},t_i)\right).
    \label{eq:sec-order-sampler-lc}
\end{align}
Here, an empty product is understood as the identity map, and the products are ordered so that the map corresponding to the earliest interval acts first.

The following theorem then shows how well this map approximates the target channel $\mathcal{E}(t,0)$.

\begin{theorem}[Second-order sampling error for the linear combination input model]
\label{thm:error-analysis-lc-2}
Let $\mathcal{L}(\tau)$ be a time-dependent Lindbladian as defined in~\eqref{eq:def-lind-lc} for all $\tau \in [0, t]$. We assume that each $\mathcal{M}_{\ell}(\tau)$ is bounded and piecewise continuous on $[0,t]$ and is nonzero everywhere. Let $\mathcal{E}(t, 0) =  \exp_{\mathcal{T}}\!\left(\int_0^t d\tau\, \mathcal{L}(\tau) \right)$ be the solution to the corresponding Lindblad master equation, and let $0=t_0<t_1<\cdots<t_r=t$ be a partition satisfying~\eqref{eq:partition-LC}. Let $\mathcal{A}^{(2,\operatorname{LC})}(t,0)$ be the second-order sampling map defined in~\eqref{eq:sec-order-sampler-lc}. Then,
\begin{equation}
    \frac{1}{2}\left\Vert\mathcal{E}(t,0)-\mathcal{A}^{(2,\operatorname{LC})}(t,0)\right\Vert_{\diamond}
    \leq
    \frac{\left\Vert\mathcal{L}\right\Vert_{\diamond,1,1}^3}{6r^2}
    + \frac{(r-1)\left\Vert\mathcal{L}\right\Vert_{\diamond,1,1}^4}{4r^3}.
\end{equation}
Consequently,
\begin{equation}
    \frac{1}{2}\left\Vert\mathcal{E}(t,0)-\mathcal{A}^{(2,\operatorname{LC})}(t,0)\right\Vert_{\diamond}
    =
    O\!\left(\frac{\left\Vert\mathcal{L}\right\Vert_{\diamond,1,1}^3+\left\Vert\mathcal{L}\right\Vert_{\diamond,1,1}^4}{r^2}\right).
\end{equation}
\end{theorem}
\begin{proof}
We first bound the local error on each time interval after including the second-order correction. 
Using $\mathcal{E}_0(t_{i+1},t_i)=\mathcal{A}^{\operatorname{(LC)}}_0(t_{i+1},t_i)=\mathcal{I}$, the first-order matching identities established in Section~\ref{sec:first-order-lc}, and the definition of $\mathcal{R}^{\operatorname{(LC)}}(t_{i+1},t_i)$ in~\eqref{eq:second-order-r-map-lc}, Taylor's formula with integral remainder gives
\begin{equation}
    \mathcal{E}(t_{i+1},t_i)-\mathcal{A}^{\operatorname{(LC)}}(t_{i+1},t_i)-\mathcal{R}^{\operatorname{(LC)}}(t_{i+1},t_i)
    = \int_0^1 ds\,\frac{(1-s)^2}{2}\frac{d^3}{ds^3}\left(\mathcal{E}_s(t_{i+1},t_i)-\mathcal{A}^{\operatorname{(LC)}}_s(t_{i+1},t_i)\right).
    \label{eq:td-second-order-local-remainder-lc}
\end{equation}

For the third-order derivative of the exact propagator, the same calculation as in~\eqref{eq:sec-o-proof-first-term-decomp-1} gives
\begin{equation}
    \left\Vert\frac{d^3}{ds^3}\mathcal{E}_s(t_{i+1},t_i)\right\Vert_{\diamond}
    \leq \left(\left\Vert\mathcal{L}\right\Vert_{\diamond,1}^{(i)}\right)^3
    \leq \left(\left\Vert\mathcal{L}\right\Vert_{\diamond,1,1}^{(i)}\right)^3.
    \label{eq:sec-o-proof-first-term-decomp-1-lc}
\end{equation}
For the linear combination sampler, differentiating under the integral sign, as in the derivation of~\eqref{eq:sec-term-bound-lc}, and using submultiplicativity of the diamond norm gives
\begin{align}
    \left\Vert\frac{d^3}{ds^3}\mathcal{A}^{\operatorname{(LC)}}_s(t_{i+1},t_i)\right\Vert_{\diamond}
    &=
    \left\Vert\frac{d^3}{ds^3}\left(\sum_{\ell=1}^{J+K}\int_{t_i}^{t_{i+1}}d\tau\,p^{(i)}(\ell,\tau)e^{s\frac{\mathcal{M}_{\ell}(\tau)}{p^{(i)}(\ell,\tau)}}\right)\right\Vert_{\diamond}\\
    &\leq
    \sum_{\ell=1}^{J+K}\int_{t_i}^{t_{i+1}}d\tau\,p^{(i)}(\ell,\tau) \left\Vert\frac{d^3}{ds^3}e^{s\frac{\mathcal{M}_{\ell}(\tau)}{p^{(i)}(\ell,\tau)}}\right\Vert_{\diamond}\\
    &\leq
    \sum_{\ell=1}^{J+K}\int_{t_i}^{t_{i+1}}d\tau\,\frac{\left\Vert\mathcal{M}_{\ell}(\tau)\right\Vert_{\diamond}^3}{\left(p^{(i)}(\ell,\tau)\right)^2}\\
    &=
    \left(\left\Vert\mathcal{L}\right\Vert_{\diamond,1,1}^{(i)}\right)^3.
    \label{eq:sec-o-proof-first-term-decomp-2-lc}
\end{align}
Applying the bounds in~\eqref{eq:sec-o-proof-first-term-decomp-1-lc} and~\eqref{eq:sec-o-proof-first-term-decomp-2-lc} to the integral remainder in~\eqref{eq:td-second-order-local-remainder-lc}, and using $\int_0^1(1-s)^2/2\,ds=1/6$, we obtain
\begin{equation}
    \left\Vert\mathcal{E}(t_{i+1},t_i)-\mathcal{A}^{\operatorname{(LC)}}(t_{i+1},t_i)-\mathcal{R}^{\operatorname{(LC)}}(t_{i+1},t_i)\right\Vert_{\diamond}
    \leq \frac{1}{3}\left(\left\Vert\mathcal{L}\right\Vert_{\diamond,1,1}^{(i)}\right)^3.
    \label{eq:td-second-order-local-residual-lc}
\end{equation}

We also require a bound on the first-order local error. Since the exact and randomized maps agree through first order in $s$, Taylor's formula gives
\begin{equation}
    \mathcal{E}(t_{i+1},t_i)-\mathcal{A}^{\operatorname{(LC)}}(t_{i+1},t_i)
    = \int_0^1 ds\,(1-s)\frac{d^2}{ds^2}\left(\mathcal{E}_s(t_{i+1},t_i)-\mathcal{A}^{\operatorname{(LC)}}_s(t_{i+1},t_i)\right).
\end{equation}
Using the second-derivative bounds in~\eqref{eq:bound-first-term-sec} and~\eqref{eq:sec-term-bound-lc}, together with $\int_0^1(1-s)\,ds=1/2$ and $\Vert\mathcal{L}\Vert_{\diamond,1}^{(i)}\leq\Vert\mathcal{L}\Vert_{\diamond,1,1}^{(i)}$, we obtain
\begin{align}
    \left\Vert\mathcal{E}(t_{i+1},t_i)-\mathcal{A}^{\operatorname{(LC)}}(t_{i+1},t_i)\right\Vert_{\diamond}
    &\leq \frac{1}{2}\left[\left(\left\Vert\mathcal{L}\right\Vert_{\diamond,1}^{(i)}\right)^2+\left(\left\Vert\mathcal{L}\right\Vert_{\diamond,1,1}^{(i)}\right)^2\right]\\
    &\leq \left(\left\Vert\mathcal{L}\right\Vert_{\diamond,1,1}^{(i)}\right)^2.
    \label{eq:td-first-order-local-bound-lc}
\end{align}
Moreover, from the definition of $\mathcal{R}^{\operatorname{(LC)}}(t_{i+1},t_i)$ in~\eqref{eq:second-order-r-map-lc} and the same second-derivative bounds,
\begin{align}
    \left\Vert\mathcal{R}^{\operatorname{(LC)}}(t_{i+1},t_i)\right\Vert_{\diamond}
    &\leq
    \frac{1}{2}\left(\left\Vert\left.\frac{d^2}{ds^2}\mathcal{E}_s(t_{i+1},t_i)\right|_{s=0}\right\Vert_{\diamond}
    +
    \left\Vert\left.\frac{d^2}{ds^2}\mathcal{A}^{\operatorname{(LC)}}_s(t_{i+1},t_i)\right|_{s=0}\right\Vert_{\diamond}\right)\\
    &\leq
    \frac{1}{2}\left[\left(\left\Vert\mathcal{L}\right\Vert_{\diamond,1}^{(i)}\right)^2
    +
    \left(\left\Vert\mathcal{L}\right\Vert_{\diamond,1,1}^{(i)}\right)^2\right]\\
    &\leq
    \left(\left\Vert\mathcal{L}\right\Vert_{\diamond,1,1}^{(i)}\right)^2.
    \label{eq:R-bound-lc}
\end{align}

We now compare the full exact evolution with the second-order sampling map. Using the composition property of the time-ordered propagator and the definition of $\mathcal{A}^{(2,\operatorname{LC})}(t,0)$ in~\eqref{eq:sec-order-sampler-lc}, a telescoping expansion of the difference between the exact and first-order products gives
\begin{align}
    &\mathcal{E}(t,0)-\mathcal{A}^{(2,\operatorname{LC})}(t,0)\notag\\
    &=
    \sum_{k=0}^{r-1}\left(\prod_{i=k+1}^{r-1}\mathcal{E}(t_{i+1},t_i)\right)\left[\mathcal{E}(t_{k+1},t_k)-\mathcal{A}^{\operatorname{(LC)}}(t_{k+1},t_k)-\mathcal{R}^{\operatorname{(LC)}}(t_{k+1},t_k)\right]\left(\prod_{i=0}^{k-1}\mathcal{A}^{\operatorname{(LC)}}(t_{i+1},t_i)\right)\notag\\
    &\quad +
    \sum_{k=0}^{r-1}\left[\prod_{i=k+1}^{r-1}\mathcal{E}(t_{i+1},t_i) - \prod_{i=k+1}^{r-1}\mathcal{A}^{\operatorname{(LC)}}(t_{i+1},t_i)\right]\mathcal{R}^{\operatorname{(LC)}}(t_{k+1},t_k)\left(\prod_{i=0}^{k-1}\mathcal{A}^{\operatorname{(LC)}}(t_{i+1},t_i)\right).
    \label{eq:td-second-order-global-decomposition-lc}
\end{align}
The first sum contains the residual local error after subtracting the second-order correction. The second sum accounts for the difference between the exact and first-order sampled evolutions following the inserted correction.

Since both $\mathcal{E}(t_{i+1},t_i)$ and $\mathcal{A}^{\operatorname{(LC)}}(t_{i+1},t_i)$ are quantum channels, all surrounding factors in the first sum have diamond norm one. Therefore, using~\eqref{eq:td-second-order-local-residual-lc}, its contribution to the normalized diamond-norm error is bounded by
\begin{equation}
    \frac{1}{2}\sum_{k=0}^{r-1}\left\Vert\mathcal{E}(t_{k+1},t_k)-\mathcal{A}^{\operatorname{(LC)}}(t_{k+1},t_k)-\mathcal{R}^{\operatorname{(LC)}}(t_{k+1},t_k)\right\Vert_{\diamond}
    \leq \frac{1}{6}\sum_{k=0}^{r-1}\left(\left\Vert\mathcal{L}\right\Vert_{\diamond,1,1}^{(k)}\right)^3.
\end{equation}

For the second sum, we again use a telescoping expansion. Since the exact and randomized maps on each interval are quantum channels,~\eqref{eq:td-first-order-local-bound-lc} gives
\begin{align}
    \left\Vert\prod_{i=k+1}^{r-1}\mathcal{E}(t_{i+1},t_i)-\prod_{i=k+1}^{r-1}\mathcal{A}^{\operatorname{(LC)}}(t_{i+1},t_i)\right\Vert_{\diamond}
    &\leq
    \sum_{j=k+1}^{r-1}\left\Vert\mathcal{E}(t_{j+1},t_j)-\mathcal{A}^{\operatorname{(LC)}}(t_{j+1},t_j)\right\Vert_{\diamond}\\
    &\leq
    \sum_{j=k+1}^{r-1}\left(\left\Vert\mathcal{L}\right\Vert_{\diamond,1,1}^{(j)}\right)^2.
\end{align}
Combining this bound with~\eqref{eq:R-bound-lc} and the decomposition in~\eqref{eq:td-second-order-global-decomposition-lc}, we obtain
\begin{equation}
    \frac{1}{2}\left\Vert\mathcal{E}(t,0)-\mathcal{A}^{(2,\operatorname{LC})}(t,0)\right\Vert_{\diamond}
    \leq
    \frac{1}{6}\sum_{k=0}^{r-1}\left(\left\Vert\mathcal{L}\right\Vert_{\diamond,1,1}^{(k)}\right)^3
    + \frac{1}{2}\sum_{k=0}^{r-1}\left[\sum_{j=k+1}^{r-1}\left(\left\Vert\mathcal{L}\right\Vert_{\diamond,1,1}^{(j)}\right)^2\right]\left(\left\Vert\mathcal{L}\right\Vert_{\diamond,1,1}^{(k)}\right)^2.
\end{equation}

Finally, by the choice of partition in~\eqref{eq:partition-LC}, $\left\Vert\mathcal{L}\right\Vert_{\diamond,1,1}^{(i)}=\left\Vert\mathcal{L}\right\Vert_{\diamond,1,1}/r$ for every $i$. Therefore, the first term becomes
\begin{equation}
    \frac{1}{6}\sum_{k=0}^{r-1}\left(\left\Vert\mathcal{L}\right\Vert_{\diamond,1,1}^{(k)}\right)^3
    = \frac{\left\Vert\mathcal{L}\right\Vert_{\diamond,1,1}^3}{6r^2},
\end{equation}
while the second term contains one contribution for each index pair $(k,j)$ satisfying $0\leq k<j\leq r-1$. There are $\binom{r}{2}=r(r-1)/2$ such pairs, and therefore
\begin{equation}
    \frac{1}{2}\sum_{k=0}^{r-1}\left[\sum_{j=k+1}^{r-1}\left(\left\Vert\mathcal{L}\right\Vert_{\diamond,1,1}^{(j)}\right)^2\right]\left(\left\Vert\mathcal{L}\right\Vert_{\diamond,1,1}^{(k)}\right)^2
    = \frac{1}{2}\binom{r}{2}\frac{\left\Vert\mathcal{L}\right\Vert_{\diamond,1,1}^4}{r^4}
    = \frac{(r-1)\left\Vert\mathcal{L}\right\Vert_{\diamond,1,1}^4}{4r^3}.
\end{equation}
Thus,
\begin{equation}
    \frac{1}{2}\left\Vert\mathcal{E}(t,0)-\mathcal{A}^{(2,\operatorname{LC})}(t,0)\right\Vert_{\diamond}
    \leq \frac{\left\Vert\mathcal{L}\right\Vert_{\diamond,1,1}^3}{6r^2}
    + \frac{(r-1)\left\Vert\mathcal{L}\right\Vert_{\diamond,1,1}^4}{4r^3}.
\end{equation}
Since $(r-1)/r\leq1$, it follows that
\begin{equation}
    \frac{1}{2}\left\Vert\mathcal{E}(t,0)-\mathcal{A}^{(2,\operatorname{LC})}(t,0)\right\Vert_{\diamond}
    =    O\!\left(\frac{\left\Vert\mathcal{L}\right\Vert_{\diamond,1,1}^3+\left\Vert\mathcal{L}\right\Vert_{\diamond,1,1}^4}{r^2}\right).
\end{equation}
This concludes the proof.
\end{proof}

\subsubsection{Implementation of the Second-Order LC Sampler}

Recall that the second-order sampling algorithm for the linear combination input model is given by~\eqref{eq:sec-order-sampler-lc}. Since the map is not necessarily CPTP, it cannot in general be implemented directly as a quantum channel. As in Section~\ref{sec:qswift-implementation}, we instead consider estimating expectation values. Let $G$ be an observable and let $\rho$ be an initial state. Our goal is to estimate
\begin{equation}
    g \coloneqq \operatorname{Tr}\!\left[G\mathcal{E}(t, 0)(\rho)\right].
\end{equation}
We define the corresponding second-order approximation as
\begin{equation}
    g^{(2)} \coloneqq \operatorname{Tr}\!\left[G\mathcal{A}^{(2,\operatorname{LC})}(t,0)(\rho)\right].
\end{equation}
From Theorem~\ref{thm:error-analysis-lc-2}, the deviation between $g$ and $g^{(2)}$ is bounded by
\begin{align}
    \left | g - g^{(2)}\right| 
    &\leq \left \Vert G \right \Vert \left\Vert \mathcal{E}(t, 0) - \mathcal{A}^{(2,\operatorname{LC})}(t,0)\right\Vert_{\diamond} \\
    &\leq \left \Vert G \right \Vert \left(\frac{\left \Vert \mathcal{L} \right \Vert_{\diamond, 1, 1}^3}{3r^2} + \frac{(r-1)\left \Vert \mathcal{L} \right \Vert_{\diamond, 1, 1}^4}{2r^3}\right).
\end{align}
Consequently,
\begin{equation}
    \left|g-g^{(2)}\right|=O\!\left(\frac{\left\Vert G\right\Vert\left(\Vert\mathcal{L}\Vert_{\diamond,1,1}^3+\Vert\mathcal{L}\Vert_{\diamond,1,1}^4\right)}{r^2}\right).
\end{equation}
Thus, choosing $r$ so that
\begin{equation}
    r=O\!\left(\sqrt{\frac{\left\Vert G\right\Vert\left(\Vert\mathcal{L}\Vert_{\diamond,1,1}^3+\Vert\mathcal{L}\Vert_{\diamond,1,1}^4\right)}{\varepsilon}}\right)
\end{equation}
makes the approximation bias at most $\varepsilon$.

In what follows, we demonstrate how to estimate $g^{(2)}$. Consider that
\begin{align}
     g^{(2)} & = \operatorname{Tr}\!\left[G\mathcal{A}^{(2,\operatorname{LC})}(t,0)(\rho)\right]\\
     & = \operatorname{Tr}\!\left[G\prod_{i=0}^{r-1}\mathcal{A}^{\operatorname{(LC)}}(t_{i+1},t_i)(\rho)
    \right]\notag\\
    &\quad +
    \sum_{k=0}^{r-1}\operatorname{Tr}\!\left[G\left(\prod_{i=k+1}^{r-1}\mathcal{A}^{\operatorname{(LC)}}(t_{i+1},t_i)\right)\mathcal{R}^{\operatorname{(LC)}}(t_{k+1},t_k)\left(\prod_{i=0}^{k-1}\mathcal{A}^{\operatorname{(LC)}}(t_{i+1},t_i)\right)(\rho)\right]\\
     & \eqqcolon g^{(1)} + \delta g 
\end{align}
Here, $g^{(1)}\coloneqq \operatorname{Tr}\!\left[G \prod_{i=0}^{r-1} \mathcal{A}^{\operatorname{(LC)}}(t_{i+1}, t_i)(\rho)\right]$ denotes the expectation value associated with the first-order continuous qDRIFT approximation, while 
\begin{equation}
    \delta g\coloneqq \sum_{k=0}^{r-1}\operatorname{Tr}\!\left[G\left(\prod_{i=k+1}^{r-1}\mathcal{A}^{\operatorname{(LC)}}(t_{i+1},t_i)\right)\mathcal{R}^{\operatorname{(LC)}}(t_{k+1},t_k)\left(\prod_{i=0}^{k-1}\mathcal{A}^{\operatorname{(LC)}}(t_{i+1},t_i)\right)(\rho)\right]
\end{equation}
denotes the second-order correction.
To estimate $g^{(1)}$, we apply the continuous qDRIFT sampling procedure corresponding to $\prod_{i=0}^{r-1} \mathcal{A}^{\operatorname{(LC)}}(t_{i+1}, t_i)$ to the initial state, measure $G$, repeat this process multiple times, and then take the sample mean. This yields an unbiased estimator of $g^{(1)}$. However, the estimation of the second term $\delta g$ is more involved and requires further decomposition.

To estimate $\delta g$, we express it as a sum of contributions, each denoted by $\delta g_k$:
\begin{equation}
    \delta g = \sum_{k=0}^{r-1} \delta g_k\label{eq:g-g_k-conti},
\end{equation}
where each term $\delta g_k$ is given by
\begin{equation}
\delta g_k \coloneqq \operatorname{Tr}\!\left[ G\left(\prod_{i=k+1}^{r-1} \mathcal{A}^{\operatorname{(LC)}}(t_{i+1}, t_i)\right) \circ \mathcal{R}^{\operatorname{(LC)}}(t_{k+1}, t_k) \circ \left( \prod_{i=0}^{k-1}   \mathcal{A}^{\operatorname{(LC)}}(t_{i+1}, t_i)\right)(\rho) \right].
\end{equation}
Using~\eqref{eq:second-order-r-map-lc}, we have
\begin{multline}
\delta g_k = \frac{1}{2}\operatorname{Tr}\!\left[ G\left(\prod_{i=k+1}^{r-1} \mathcal{A}^{\operatorname{(LC)}}(t_{i+1}, t_i)\right) \circ \left. \frac{d^2}{ds^2} \mathcal{E}_{s}(t_{k+1}, t_k) \right|_{s=0} \circ \left( \prod_{i=0}^{k-1}   \mathcal{A}^{\operatorname{(LC)}}(t_{i+1}, t_i)\right)(\rho) \right]\\
- \frac{1}{2}\operatorname{Tr}\!\left[ G\left(\prod_{i=k+1}^{r-1} \mathcal{A}^{\operatorname{(LC)}}(t_{i+1}, t_i)\right) \circ \left. \frac{d^2}{ds^2} \mathcal{A}^{\operatorname{(LC)}}_{s}(t_{k+1}, t_k) \right|_{s=0} \circ \left( \prod_{i=0}^{k-1}   \mathcal{A}^{\operatorname{(LC)}}(t_{i+1}, t_i)\right)(\rho) \right]\label{eq:conti-delta-gk-Rlc}.
\end{multline}

We now expand the two second derivative terms appearing in~\eqref{eq:conti-delta-gk-Rlc}. First, using the derivative identity for the scaled propagator from Lemma~\ref{lem:time-dep-deriv-for-err-bnd}, we have
\begin{equation}
    \frac{d}{ds}\mathcal{E}_{s}(b,a)=\int_a^b d\tau\,\mathcal{E}_{s}(b,\tau)\circ\mathcal{L}(\tau)\circ\mathcal{E}_{s}(\tau,a).
    \label{eq:propagator-deriv-recall-lc}
\end{equation}
Applying this identity once more and then setting $s=0$, where $\mathcal{E}_0(\cdot,\cdot)=\mathcal{I}$, gives
\begin{align}
    & \left. \frac{d^2}{ds^2} \mathcal{E}_{s}(t_{k+1}, t_k) \right|_{s=0} \nonumber\\
    & =  \int_{t_{k}}^{t_{k+1}} d\tau_1\ \left[\left.\frac{d}{ds} \mathcal{E}_{s}(t_{k+1}, \tau_1)\right|_{s=0} \circ \mathcal{L}(\tau_1)  +  \mathcal{L}(\tau_1) \circ \left.\frac{d}{ds} \mathcal{E}_{s}(\tau_1, t_{k}) \right|_{s=0}\right] \\
    & = \int_{t_{k}}^{t_{k+1}} d\tau_1\ \left(  \int_{\tau_1}^{t_{k+1}} d\tau_2\ \mathcal{L}(\tau_2) \circ \mathcal{L}(\tau_1)  +  \int_{t_{k}}^{\tau_1} d\tau_2\   \mathcal{L}(\tau_1) \circ \mathcal{L}(\tau_2) \right).
\end{align}
The two terms correspond to the two possible time orderings of the pair $(\tau_1,\tau_2)$.
Equivalently, we can rewrite the same expression as an integral over the full square $[t_k,t_{k+1}]^2$:
\begin{align}
    &\left.\frac{d^2}{ds^2}\mathcal{E}_{s}(t_{k+1},t_k)\right|_{s=0}\notag\\
    &=\int_{t_k}^{t_{k+1}} d\tau_1\int_{t_k}^{t_{k+1}} d\tau_2\left(\mathcal{L}(\tau_1)\circ\mathcal{L}(\tau_2)\mathbf{1}_{\tau_2\leq\tau_1}+\mathcal{L}(\tau_2)\circ\mathcal{L}(\tau_1)\mathbf{1}_{\tau_2\geq\tau_1}\right).
    \label{eq:E-second-deriv-indicator-lc}
\end{align}
Here $\mathbf{1}_{\tau_2\leq\tau_1}$ is the indicator function that equals one when $\tau_2\leq\tau_1$ and zero otherwise, and similarly for $\mathbf{1}_{\tau_2\geq\tau_1}$. 
The diagonal $\tau_1=\tau_2$ has measure zero and therefore does not affect the value of the double integral.

We now rewrite~\eqref{eq:E-second-deriv-indicator-lc} in terms of the LC sampling distribution. 
Recall from~\eqref{eq:def-lind-lc} that $\mathcal{L}(\tau) = \sum_{\ell=1}^{J+K}\mathcal{M}_{\ell}(\tau)$.
For each nonzero local term, define the normalized generator
\begin{equation}
    \overline{\mathcal{M}}_{\ell}(\tau)\coloneqq\frac{\mathcal{M}_{\ell}(\tau)}{\left\Vert\mathcal{M}_{\ell}(\tau)\right\Vert_{\diamond}}.
    \label{eq:normalized-M-lc}
\end{equation}
Also recall from~\eqref{eq:prob-den-lc-i} that, on the $k$-th interval $[t_k,t_{k+1}]$, the joint probability density over $(\ell,\tau)$ is $p^{(k)}(\ell,\tau)=\left\Vert\mathcal{M}_{\ell}(\tau)\right\Vert_{\diamond}/\left\Vert\mathcal{L}\right\Vert_{\diamond,1,1}^{(k)}$.
Equivalently, for $\tau\in[t_k,t_{k+1}]$,
\begin{equation}
    \mathcal{M}_{\ell}(\tau)=\left\Vert\mathcal{L}\right\Vert_{\diamond,1,1}^{(k)}p^{(k)}(\ell,\tau)\overline{\mathcal{M}}_{\ell}(\tau).
    \label{eq:M-factorization-lc}
\end{equation}
Substituting this expression into~\eqref{eq:E-second-deriv-indicator-lc}, we obtain
\begin{align}
    & \left. \frac{d^2}{ds^2} \mathcal{E}_{s}(t_{k+1}, t_k) \right|_{s=0} \nonumber\\
    & = \left(\left \Vert \mathcal{L} \right \Vert_{\diamond, 1, 1}^{(k)}\right)^2 \Bigg(\sum_{\ell_1, \ell_2 = 1}^{J+K} \int_{t_{k}}^{t_{k+1}} d\tau_1\, p^{(k)}(\ell_1,\tau_1)  \int_{t_{k}}^{t_{k+1}} d\tau_2\, p^{(k)}(\ell_2,\tau_2) \Bigg(\overline{\mathcal{M}}_{\ell_1}(\tau_1) \circ \overline{\mathcal{M}}_{\ell_2}(\tau_2) \mathbf{1}_{\tau_2 \leq \tau_1} \nonumber\\
    & \qquad +  \overline{\mathcal{M}}_{\ell_2}(\tau_2) \circ  \overline{\mathcal{M}}_{\ell_1}(\tau_1) \mathbf{1}_{\tau_2 \geq \tau_1} \Bigg) \Bigg)\\
    & = \frac{\left \Vert \mathcal{L} \right \Vert_{\diamond, 1, 1}^2}{r^2} \Bigg(\sum_{\ell_1, \ell_2 = 1}^{J+K} \int_{t_{k}}^{t_{k+1}} d\tau_1\, p^{(k)}(\ell_1,\tau_1)  \int_{t_{k}}^{t_{k+1}} d\tau_2\, p^{(k)}(\ell_2,\tau_2) \Bigg(\overline{\mathcal{M}}_{\ell_1}(\tau_1) \circ \overline{\mathcal{M}}_{\ell_2}(\tau_2) \mathbf{1}_{\tau_2 \leq \tau_1} \nonumber\\
    & \qquad +  \overline{\mathcal{M}}_{\ell_2}(\tau_2) \circ  \overline{\mathcal{M}}_{\ell_1}(\tau_1) \mathbf{1}_{\tau_2 \geq \tau_1} \Bigg) \Bigg),
\end{align}
where the last equality follows from the partition condition~\eqref{eq:partition-LC}.

We next expand the corresponding second derivative of the LC sampler. From the definition of
$\mathcal{A}^{\operatorname{(LC)}}_s(t_{k+1},t_k)$, we have
\begin{align}
 \left. \frac{d^2}{ds^2} \mathcal{A}^{\operatorname{(LC)}}_{s}(t_{k+1}, t_k) \right|_{s=0}& =  \left(\left. \sum_{\ell = 1}^{J+K} \int_{t_k}^{t_{k+1}} d\tau\ p^{(k)}(\ell,\tau)\  \frac{\mathcal{M}_{\ell}(\tau)}{p^{(k)}(\ell,\tau)} \circ \frac{\mathcal{M}_{\ell}(\tau)}{p^{(k)}(\ell,\tau)}  \circ
 e^{s\frac{\mathcal{M}_{\ell}(\tau)}{p^{(k)}(\ell,\tau)}}\right) \right|_{s=0}\\
 & =  \sum_{\ell = 1}^{J+K} \int_{t_k}^{t_{k+1}} d\tau\ p^{(k)}(\ell,\tau)\  \frac{\mathcal{M}_{\ell}(\tau)}{p^{(k)}(\ell,\tau)} \circ \frac{\mathcal{M}_{\ell}(\tau)}{p^{(k)}(\ell,\tau)}\\
 & = \left(\left \Vert \mathcal{L} \right \Vert_{\diamond, 1, 1}^{(k)}\right)^2 \sum_{\ell = 1}^{J+K} \int_{t_k}^{t_{k+1}} d\tau\ p^{(k)}(\ell,\tau)\,  \overline{\mathcal{M}}_{\ell}(\tau) \circ \overline{\mathcal{M}}_{\ell}(\tau)\\
 & = \frac{\left \Vert \mathcal{L} \right \Vert_{\diamond, 1, 1}^2}{r^2} \sum_{\ell = 1}^{J+K} \int_{t_k}^{t_{k+1}} d\tau\ p^{(k)}(\ell,\tau)\  \overline{\mathcal{M}}_{\ell}(\tau) \circ \overline{\mathcal{M}}_{\ell}(\tau).
\end{align}

Using the above and simplifying it further, we can rewrite~\eqref{eq:conti-delta-gk-Rlc} as follows:
\begin{align}
    \delta g_k & = \frac{\left \Vert \mathcal{L} \right \Vert_{\diamond, 1, 1}^2}{2r^2} \Bigg(\sum_{\ell_1, \ell_2 = 1}^{J+K} \int_{t_{k}}^{t_{k+1}} d\tau_1\, p^{(k)}(\ell_1,\tau_1)  \int_{t_{k}}^{t_{k+1}} d\tau_2\, p^{(k)}(\ell_2,\tau_2) \nonumber\\ 
    & \qquad \times\operatorname{Tr}\!\Bigg[ G\left(\prod_{i=k+1}^{r-1} \mathcal{A}^{\operatorname{(LC)}}(t_{i+1}, t_i)\right) \circ \Bigg(\overline{\mathcal{M}}_{\ell_1}(\tau_1) \circ  \overline{\mathcal{M}}_{\ell_2}(\tau_2)\ \mathbf{1}_{\tau_2 \leq \tau_1} \nonumber\\
      & \qquad +  \overline{\mathcal{M}}_{\ell_2}(\tau_2) \circ  \overline{\mathcal{M}}_{\ell_1}(\tau_1)\  \mathbf{1}_{\tau_2 \geq \tau_1} \Bigg)  \circ  \left( \prod_{i=0}^{k-1}   \mathcal{A}^{\operatorname{(LC)}}(t_{i+1}, t_i)\right)(\rho) \Bigg] \nonumber\\
      & \qquad - \sum_{\ell = 1}^{J+K} \int_{t_k}^{t_{k+1}} d\tau\ p^{(k)}(\ell,\tau)\ \operatorname{Tr}\!\Bigg[ G\left(\prod_{i=k+1}^{r-1} \mathcal{A}^{\operatorname{(LC)}}(t_{i+1}, t_i)\right) \circ \overline{\mathcal{M}}_{\ell}(\tau) \circ \overline{\mathcal{M}}_{\ell}(\tau) \nonumber\\
      & \qquad \circ\left( \prod_{i=0}^{k-1}   \mathcal{A}^{\operatorname{(LC)}}(t_{i+1}, t_i)\right)(\rho) \Bigg]\Bigg).\label{eq:delta-gk-expanded-lc}
\end{align}

Thus, the basic quantity we need to estimate has the form
\begin{equation}
    \operatorname{Tr}\!\left[ G\left(\prod_{i=k+1}^{r-1} \mathcal{A}^{\operatorname{(LC)}}(t_{i+1}, t_i)\right) \circ \overline{\mathcal{M}}_{\ell_1}(\tau_1) \circ \overline{\mathcal{M}}_{\ell_2}(\tau_2) \circ \left( \prod_{i=0}^{k-1}   \mathcal{A}^{\operatorname{(LC)}}(t_{i+1}, t_i)\right)(\rho) \right].
    \label{eq:building-block-conti}
\end{equation}
This can be estimated using the same Hadamard-test-type construction as in Section~\ref{sec:qswift-implementation}. 
We initialize a qSWIFT control qubit in $|+\rangle=(|0\rangle+|1\rangle)/\sqrt{2}$, so that
\begin{equation}
    |+\rangle\!\langle+| \otimes \rho
    = \frac{1}{2}
    \begin{pmatrix}
        \rho & \rho \\
        \rho & \rho
    \end{pmatrix}.
\end{equation}
As in Section~\ref{sec:qswift-implementation}, the diagonal blocks will not be relevant for the final $X\otimes G$ measurement, and we display only the action on the off-diagonal blocks.
The full normalized swift maps on arbitrary joint control-system operators will be defined below.
The normalized swift maps below are obtained from the swift maps in Section~\ref{sec:qswift-implementation} by including the additional factor $1/\left\Vert\mathcal{M}_{\ell}(\tau)\right\Vert_{\diamond}$, so that their sum realizes $\overline{\mathcal{M}}_{\ell}(\tau)$ on the off-diagonal blocks.

For a Hamiltonian generator $\mathcal{M}_{\ell}(\tau)(\sigma) = -i[H_{\ell}(\tau),\sigma]$, let $\mathcal{B}_{\ell}=\{0,1\}$. The corresponding normalized swift maps act as
\begin{align}
    \begin{pmatrix}
        \cdot & \sigma \\
        \sigma & \cdot
    \end{pmatrix}
    &\xrightarrow{\overline{\mathcal{S}}^{(0)}_{\ell}(\tau)} \frac{1}{\left\Vert\mathcal{M}_{\ell}(\tau)\right\Vert_{\diamond}}
    \begin{pmatrix}
        \cdot & i\sigma H_{\ell}(\tau) \\
        -iH_{\ell}(\tau)\sigma & \cdot
    \end{pmatrix}, \\[0.3em]
    \begin{pmatrix}
        \cdot & \sigma \\
        \sigma & \cdot
    \end{pmatrix}
    &\xrightarrow{\overline{\mathcal{S}}^{(1)}_{\ell}(\tau)} \frac{1}{\left\Vert\mathcal{M}_{\ell}(\tau)\right\Vert_{\diamond}}
    \begin{pmatrix}
        \cdot & -iH_{\ell}(\tau)\sigma \\
        i\sigma H_{\ell}(\tau) & \cdot
    \end{pmatrix}.
\end{align}
Therefore,
\begin{equation}
    \begin{pmatrix}
        \cdot & \sigma \\
        \sigma & \cdot
    \end{pmatrix}
    \xrightarrow{
    \overline{\mathcal{S}}^{(0)}_{\ell}(\tau) + \overline{\mathcal{S}}^{(1)}_{\ell}(\tau) }
    \begin{pmatrix}
        \cdot &
        \overline{\mathcal{M}}_{\ell}(\tau)(\sigma) \\
        \overline{\mathcal{M}}_{\ell}(\tau)(\sigma) &
        \cdot
    \end{pmatrix}.
    \label{eq:normalized-swift-H-lc}
\end{equation}

For a Lindbladian generator $\mathcal{M}_{\ell}(\tau)(\sigma)=L_{\ell-J}(\tau)\sigma L_{\ell-J}^{\dagger}(\tau)-\tfrac{1}{2}\{L_{\ell-J}^{\dagger}(\tau)L_{\ell-J}(\tau),\sigma\}$, let $\mathcal{B}_{\ell}=\{2,3,4\}$. The normalized swift maps act as
\begin{align}
    \begin{pmatrix}
        \cdot & \sigma \\
        \sigma & \cdot
    \end{pmatrix}
    &\xrightarrow{\overline{\mathcal{S}}^{(2)}_{\ell}(\tau)}
    \frac{1}{\left\Vert\mathcal{M}_{\ell}(\tau)\right\Vert_{\diamond}}
    \begin{pmatrix}
        \cdot & L_{\ell-J}(\tau)\sigma L_{\ell-J}^{\dagger}(\tau) \\
        L_{\ell-J}(\tau)\sigma L_{\ell-J}^{\dagger}(\tau) & \cdot
    \end{pmatrix}, \\[0.3em]
    \begin{pmatrix}
        \cdot & \sigma \\
        \sigma & \cdot
    \end{pmatrix}
    &\xrightarrow{\overline{\mathcal{S}}^{(3)}_{\ell}(\tau)}
    \frac{1}{\left\Vert\mathcal{M}_{\ell}(\tau)\right\Vert_{\diamond}}
    \begin{pmatrix}
        \cdot & -\tfrac12 L_{\ell-J}^{\dagger}(\tau)L_{\ell-J}(\tau)\sigma \\
        -\tfrac12 \sigma L_{\ell-J}^{\dagger}(\tau)L_{\ell-J}(\tau) & \cdot
    \end{pmatrix}, \\[0.3em]
    \begin{pmatrix}
        \cdot & \sigma \\
        \sigma & \cdot
    \end{pmatrix}
    &\xrightarrow{\overline{\mathcal{S}}^{(4)}_{\ell}(\tau)}
    \frac{1}{\left\Vert\mathcal{M}_{\ell}(\tau)\right\Vert_{\diamond}}
    \begin{pmatrix}
        \cdot & -\tfrac12 \sigma L_{\ell-J}^{\dagger}(\tau)L_{\ell-J}(\tau) \\
        -\tfrac12 L_{\ell-J}^{\dagger}(\tau)L_{\ell-J}(\tau)\sigma & \cdot
    \end{pmatrix}.
\end{align}
Summing these maps gives
\begin{equation}
    \begin{pmatrix}
        \cdot & \sigma \\
        \sigma & \cdot
    \end{pmatrix}
    \xrightarrow{
    \sum_{b\in\mathcal{B}_{\ell}}
    \overline{\mathcal{S}}^{(b)}_{\ell}(\tau)
    }
    \begin{pmatrix}
        \cdot &
        \overline{\mathcal{M}}_{\ell}(\tau)(\sigma) \\
        \overline{\mathcal{M}}_{\ell}(\tau)(\sigma) &
        \cdot
    \end{pmatrix}.
    \label{eq:normalized-swift-D-lc}
\end{equation}

We next describe how the normalized swift maps can be implemented. 
Using the block operator $\mathsf{K}(A,B)$ defined in~\eqref{eq:swift-block-operator}, the normalized Hamiltonian swift maps correspond to
\begin{align}
    \mathsf{K}_{\ell}^{(0)}(\tau)&\coloneqq
    \mathsf{K}\!\left(I,-\frac{iH_{\ell}(\tau)}{\left\Vert\mathcal{M}_{\ell}(\tau)\right\Vert_{\diamond}}\right),\\
    \mathsf{K}_{\ell}^{(1)}(\tau)&\coloneqq
    \mathsf{K}\!\left(-\frac{iH_{\ell}(\tau)}{\left\Vert\mathcal{M}_{\ell}(\tau)\right\Vert_{\diamond}}, I\right).
\end{align}
For the normalized Lindbladian swift maps, writing $L=L_{\ell-J}(\tau)$ for brevity, we similarly have
\begin{align}
    \mathsf{K}_{\ell}^{(2)}(\tau)&\coloneqq
    \mathsf{K}\!\left(\frac{L}{\sqrt{\left\Vert\mathcal{M}_{\ell}(\tau)\right\Vert_{\diamond}}},\frac{L}{\sqrt{\left\Vert\mathcal{M}_{\ell}(\tau)\right\Vert_{\diamond}}}
    \right),\\
    \mathsf{K}_{\ell}^{(3)}(\tau)&\coloneqq
    \mathsf{K}\!\left(-\frac{L^\dagger L}{2\left\Vert\mathcal{M}_{\ell}(\tau)\right\Vert_{\diamond}}, I\right),\\
    \mathsf{K}_{\ell}^{(4)}(\tau)&\coloneqq
    \mathsf{K}\!\left(I,-\frac{L^\dagger L}{2\left\Vert\mathcal{M}_{\ell}(\tau)\right\Vert_{\diamond}}\right).
\end{align}
The corresponding normalized swift maps are defined on an arbitrary joint control-system operator $\Omega$ by
\begin{equation}
    \overline{\mathcal{S}}_{\ell}^{(b)}(\tau)(\Omega)
    \coloneqq
    \mathsf{K}_{\ell}^{(b)}(\tau)\Omega\mathsf{K}_{\ell}^{(b)}(\tau)^\dagger.
\end{equation}
These definitions reproduce the normalized swift maps defined above on the off-diagonal blocks.

Since $H_{\ell}(\tau)$ and $L_{\ell-J}(\tau)$ are local, the corresponding operators $\mathsf{K}_{\ell}^{(b)}(\tau)$ also act only on the qSWIFT control qubit and a constant-size subsystem.
However, the normalization by $\left\Vert\mathcal{M}_{\ell}(\tau)\right\Vert_{\diamond}$ does not in general guarantee that $\mathsf{K}_{\ell}^{(b)}(\tau)$ has operator norm at most one.
We therefore choose a normalization factor $\alpha_{\ell}^{(b)}(\tau)\geq\left\Vert \mathsf{K}_{\ell}^{(b)}(\tau)\right\Vert$.
For the observable estimator below, we additionally assume that there exists a finite constant $\alpha_{\ast}$ such that
\begin{equation}
    \left\Vert \mathsf{K}_{\ell}^{(b)}(\tau)\right\Vert
    \leq
    \alpha_{\ell}^{(b)}(\tau)
    \leq
    \alpha_{\ast}
    <\infty
\end{equation}
for almost every sampled pair $(\ell,\tau)$ with respect to the joint distribution $p^{(k)}(\ell,\tau)$ and every $b\in\mathcal{B}_{\ell}$ selected with nonzero probability, uniformly over all intervals $k$.
This is an additional implementation assumption and does not follow from boundedness and piecewise continuity of the generators.

Then $\mathsf{K}_{\ell}^{(b)}(\tau)/\alpha_{\ell}^{(b)}(\tau)$ is a contraction, so the unitary dilation in~\eqref{eq:unitary-dilation-swift} can be applied with $C=\mathsf{K}_{\ell}^{(b)}(\tau)/\alpha_{\ell}^{(b)}(\tau)$.
The $|0\rangle$ heralding outcome therefore realizes the corresponding normalized swift-map contribution scaled by $1/\left(\alpha_{\ell}^{(b)}(\tau)\right)^2$.
Accordingly, a successful measurement outcome is multiplied by $\left(\alpha_{\ell}^{(b)}(\tau)\right)^2$, while a failed heralding outcome is assigned the value zero.

We now use these normalized swift maps to estimate the building block in~\eqref{eq:building-block-conti}. 
For any fixed tuple $(k,\ell_1,\tau_1,\ell_2,\tau_2)$, the contribution involving $\overline{\mathcal{M}}_{\ell_1}(\tau_1)\circ \overline{\mathcal{M}}_{\ell_2}(\tau_2)$ can be written as
\begin{align}
    & \operatorname{Tr}\!\left[ G\left(\prod_{i=k+1}^{r-1} \mathcal{A}^{\operatorname{(LC)}}(t_{i+1}, t_i)\right) \circ \overline{\mathcal{M}}_{\ell_1}(\tau_1) \circ \overline{\mathcal{M}}_{\ell_2}(\tau_2) \circ \left( \prod_{i=0}^{k-1}   \mathcal{A}^{\operatorname{(LC)}}(t_{i+1}, t_i)\right)(\rho) \right]\nonumber \\
    & = \sum_{\substack{b_1 \in \mathcal{B}_{\ell_1}\\b_2 \in \mathcal{B}_{\ell_2}}
    } \operatorname{Tr}\!\Bigg[( X\otimes G)\left( \mathrm{id} \otimes \prod_{i=k+1}^{r-1} \mathcal{A}^{\operatorname{(LC)}}(t_{i+1}, t_i)\right)\circ\overline{\mathcal{S}}_{\ell_1}^{(b_1)}(\tau_1) \circ \overline{\mathcal{S}}_{\ell_2}^{(b_2)}(\tau_2) \nonumber\\
    & \qquad \qquad \circ\left( \mathrm{id} \otimes \prod_{i=0}^{k-1}   \mathcal{A}^{\operatorname{(LC)}}(t_{i+1}, t_i)\right)\left(|+\rangle\!\langle+|\otimes \rho\right) \Bigg].\label{eq:building-block-swift-lc}
\end{align}
Here the products before and after the normalized swift maps act only on the system register, while the normalized swift maps act on the joint qSWIFT control qubit and system register. The contribution involving the opposite ordering $\overline{\mathcal{M}}_{\ell_2}(\tau_2)\circ \overline{\mathcal{M}}_{\ell_1}(\tau_1)$ is obtained by reversing the order of the two normalized swift maps:
\[
    \overline{\mathcal{S}}_{\ell_2}^{(b_2)}(\tau_2) 
    \circ \overline{\mathcal{S}}_{\ell_1}^{(b_1)}(\tau_1).
\]
Similarly, the contribution involving $\overline{\mathcal{M}}_{\ell}(\tau)\circ \overline{\mathcal{M}}_{\ell}(\tau)$ is obtained by setting $(\ell_1,\tau_1)=(\ell_2,\tau_2)=(\ell,\tau)$ in~\eqref{eq:building-block-swift-lc}.

From~\eqref{eq:delta-gk-expanded-lc} and~\eqref{eq:building-block-swift-lc}, we can construct an unbiased estimator of $\delta g$. 
For each $\ell$, define $p_b^{(\ell)}\coloneqq 1/|\mathcal{B}_{\ell}|$ for $b\in\mathcal{B}_{\ell}$.
We use a binary variable $\zeta\in\{0,1\}$ to sample between the two contributions in~\eqref{eq:delta-gk-expanded-lc}.

\medskip
\noindent
\textbf{Unbiased Estimator for $\delta g$.}  
To estimate $\delta g$, repeat the following procedure independently several times and take the sample mean of the recorded outcomes:
\begin{enumerate}
    \item Sample $\zeta\in\{0,1\}$ uniformly and sample $k\in\{0,\ldots,r-1\}$ uniformly.
    \item If $\zeta=0$, sample $(\ell_1,\tau_1)$ and $(\ell_2,\tau_2)$ independently from the joint distribution $p^{(k)}(\ell,\tau)$. 
    If $\zeta=1$, sample $(\ell_1,\tau_1)$ from $p^{(k)}(\ell,\tau)$ and set $(\ell_2,\tau_2)=(\ell_1,\tau_1)$.
    \item Sample $b_1\in\mathcal{B}_{\ell_1}$ and $b_2\in\mathcal{B}_{\ell_2}$ uniformly. 
    Denote the corresponding probabilities by $p_{b_1}^{(\ell_1)}$ and $p_{b_2}^{(\ell_2)}$.
    \item Prepare the joint initial state $|+\rangle\!\langle+|\otimes\rho$, where the qSWIFT control qubit is initialized in $|+\rangle$ and the system in $\rho$.
    \item Apply the continuous qDRIFT sampling procedure over the intervals $i=0,\ldots,k-1$, so that the averaged map is $\prod_{i=0}^{k-1}\mathcal{A}^{\operatorname{(LC)}}(t_{i+1},t_i)$.
    \item Apply the two selected normalized swift-map primitives in time order. 
    If $\tau_2\leq\tau_1$, apply $\overline{\mathcal{S}}_{\ell_2}^{(b_2)}(\tau_2)$ followed by $\overline{\mathcal{S}}_{\ell_1}^{(b_1)}(\tau_1)$. 
    If $\tau_2>\tau_1$, apply $\overline{\mathcal{S}}_{\ell_1}^{(b_1)}(\tau_1)$ followed by $\overline{\mathcal{S}}_{\ell_2}^{(b_2)}(\tau_2)$. 
    For each normalized swift-map primitive, initialize the heralding qubit in $|0\rangle$, apply the unitary dilation corresponding to $\mathsf{K}_{\ell}^{(b)}(\tau)/\alpha_{\ell}^{(b)}(\tau)$, and measure the heralding qubit.
    \item If any heralding measurement fails, record $0$ for this shot and retain it in the sample mean.
    Otherwise, apply the continuous qDRIFT sampling procedure over the intervals $i=k+1,\ldots,r-1$, so that the averaged map is $\prod_{i=k+1}^{r-1}\mathcal{A}^{\operatorname{(LC)}}(t_{i+1},t_i)$, and measure $X\otimes G$, with $X$ acting on the qSWIFT control qubit and $G$ acting on the system.
    \item For a successful shot, multiply the measurement outcome by
    \[
        \frac{(-1)^{\zeta}\left(\alpha_{\ell_1}^{(b_1)}(\tau_1)\right)^2\left(\alpha_{\ell_2}^{(b_2)}(\tau_2)\right)^2}
        {p_{b_1}^{(\ell_1)}p_{b_2}^{(\ell_2)}},
    \]
    scale it by $\left\Vert\mathcal{L}\right\Vert_{\diamond,1,1}^2/r$, and record the result.
\end{enumerate}
Let $Y$ denote the recorded correction outcome from a single trial of the above procedure. 
Equivalently, after incorporating the sampling weights and the heralding rescaling, the recorded outcome can be written as
\begin{equation}
    Y = \frac{\left\Vert\mathcal{L}\right\Vert_{\diamond,1,1}^{2}}{r} \frac{(-1)^\zeta\left(\alpha_{\ell_1}^{(b_1)}(\tau_1)\right)^2\left(\alpha_{\ell_2}^{(b_2)}(\tau_2)\right)^2}{p_{b_1}^{(\ell_1)}p_{b_2}^{(\ell_2)}}Z,
\end{equation}
where $Z$ denotes the final $X\otimes G$ measurement outcome for a successful heralding event and is set to zero if any heralding measurement fails.
The factor $\left\Vert\mathcal{L}\right\Vert_{\diamond,1,1}^{2}/r$ accounts for the uniform sampling of $\zeta\in\{0,1\}$ and $k\in\{0,\ldots,r-1\}$ together with the prefactor in~\eqref{eq:delta-gk-expanded-lc}, while the factors $1/p_{b_j}^{(\ell_j)}$ and $\left(\alpha_{\ell_j}^{(b_j)}(\tau_j)\right)^2$ respectively undo the component sampling and heralding rescaling.

Under the bounded-normalization assumption above, the magnitude of each recorded correction outcome is bounded by
\begin{equation}
    |Y| \leq \frac{9\alpha_{\ast}^{4}\left\Vert\mathcal{L}\right\Vert_{\diamond,1,1}^{2}\left\Vert G\right\Vert}{r}.
\end{equation}
Indeed, $1/p_b^{(\ell)}=|\mathcal{B}_{\ell}|\leq3$, each normalization factor is bounded by $\alpha_{\ast}$, and $|Z|\leq\left\Vert G\right\Vert$.
Thus the recorded outcome is integrable, and the fixed-parameter identities above can be averaged over the algorithmic sampling.
Under the stated sampling and bounded-normalization assumptions, the expectation value of the recorded outcome is $\delta g$, and hence the sample mean gives an unbiased estimator of $\delta g$. 
Combining this with the sample-mean estimator of $g^{(1)}$ gives an unbiased estimator of $g^{(2)}$.

\section{Conclusion}

In this work, we developed randomized sampling algorithms for Lindbladian simulation in both the time-independent and time-dependent settings. 
For time-independent Lindbladians with a finite local decomposition, we introduced a qSWIFT-type second-order correction to the averaged qDRIFT channel by matching the leading terms in the one-step expansions of the exact Lindbladian evolution and the randomized approximation. 
We proved that the resulting approximation has normalized diamond-norm error $O((\lambda^3t^3+\lambda^4t^4)/r^2)$ under the stated conditions. 
In particular, when $\lambda t\geq 1$, a normalized diamond-norm error of at most $\varepsilon$ can be achieved using $r=O(\lambda^2t^2/\sqrt{\varepsilon})$ time slices, improving the $\varepsilon^{-1}$ precision dependence of the first-order qDRIFT-type approximation to $\varepsilon^{-1/2}$.

For general time-dependent Lindbladians, we introduced a continuous-time randomized sampler whose averaged channel matches the first-order term of the time-ordered evolution, with an error bound determined by the integrated diamond norm of the Lindbladian. 
We then constructed a second-order correction from the difference between the second derivatives of the exact time-ordered propagator and the randomized one-step channel. 
The resulting second-order approximation achieves $O(r^{-2})$ convergence in normalized diamond norm and therefore an $\varepsilon^{-1/2}$ dependence on the target precision in the number of time slices. 
We further extended both the first- and second-order constructions to a linear-combination input model over local Hamiltonian and dissipative generators, allowing the randomized sampler to be formulated in terms of sampled local evolutions rather than direct implementation of the full time-dependent Lindbladian.

Because the second-order corrected maps are not necessarily CPTP, we also considered their implementation through observable estimation rather than deterministic state preparation. 
Under the stated implementation assumptions, using swift-map decompositions together with local unitary dilations and heralded implementations when needed, we constructed unbiased estimators of the corresponding second-order expectation values. 
This provides an operational way to use the second-order corrections without requiring the corrected maps themselves to be physical quantum channels.

The present results suggest several directions for further development. 
The expansion underlying the time-independent construction naturally suggests higher-order qSWIFT-type corrections beyond the second-order scheme analyzed here. 
Another important direction is to reduce the statistical overhead of the observable estimation procedures, for example by optimizing the sampling of correction terms and the corresponding swift map contributions. 
A sharper analysis of the variance and sample complexity of these estimators would further characterize the overall resource requirements of the second-order randomized algorithms.

\section*{Acknowledgements}
This work was supported by the National Research Foundation of Korea (NRF) Grants (No.~RS-2024-00438415). HK is supported by the KIAS Individual Grant No.~CG085302 at Korea Institute for Advanced Study. DP and MMW acknowledge support from the Air Force Office of Scientific Research under agreement no.~FA2386-24-1-4069.    
    
The U.S.~Government is authorized to reproduce and distribute reprints for Governmental purposes notwithstanding any copyright notation thereon. The views and conclusions contained herein are those of the authors and should not be interpreted as necessarily representing the official policies or endorsements, either expressed or implied, of the United States Air Force.

\section*{Declaration of AI use}
During the final preparation of this manuscript, the authors used ChatGPT-6 Astra to assist with proofreading and consistency checks. All suggested revisions were reviewed and edited as appropriate by the authors, who take full responsibility for the content of the manuscript.

\bibliography{ref}
\bibliographystyle{alphaurl}

\appendix

\section{Proof of Proposition~\ref{prop:qdrift-t-algo-gc}}\label{app:qdrift-t-algo-gc}
For clarity, we suppress system and reference-system labels throughout the proof. 
We begin by comparing one step of the exact Lindbladian evolution with the averaged one-step qDRIFT channel.

Defining the superoperator
\begin{equation}
    \mathcal{H}_j(\omega)\coloneqq -i[H_j,\omega],
\end{equation}
the Lindbladian can be written as
\begin{equation}
    \mathcal{L} = \lambda
    \left(\sum_{j=1}^{J}p_j\mathcal{H}_j + \sum_{k=1}^{K}p_{J+k}\mathcal{L}_k\right),
\end{equation}
while the averaged one-step qDRIFT channel is
\begin{equation}
    \mathcal{Q}_{\tau} = \sum_{j=1}^{J} p_j e^{\lambda\tau\mathcal{H}_j} + \sum_{k=1}^{K}p_{J+k}e^{\lambda\tau\mathcal{L}_k}.
\end{equation}

Separating the constant and linear terms from each exponential, using the normalization of the sampling probabilities and the expression for $\mathcal{L}$ above, we can therefore rewrite the averaged qDRIFT channel as
\begin{equation}
    \mathcal{Q}_{\tau} =
    \mathcal{I} + \tau\mathcal{L} 
    + \sum_{j=1}^{J}p_j\left(e^{\lambda\tau\mathcal{H}_j}-\mathcal{I}-\lambda\tau\mathcal{H}_j\right)
    + \sum_{k=1}^{K}p_{J+k}\left(e^{\lambda\tau\mathcal{L}_k}-\mathcal{I}-\lambda\tau\mathcal{L}_k\right).
\end{equation}
Thus, the constant term $\mathcal{I}$ and the term linear in $\tau$, namely $\tau\mathcal{L}$, agree with those of the exact one-step evolution $e^{\tau\mathcal{L}}$.
Subtracting $\mathcal{Q}_{\tau}$ from $e^{\tau\mathcal{L}}$ therefore gives
\begin{equation}
    e^{\tau\mathcal{L}}-\mathcal{Q}_{\tau} =
    \left(e^{\tau\mathcal{L}}-\mathcal{I}-\tau\mathcal{L}\right)
    -\sum_{j=1}^{J}p_j\left(e^{\lambda\tau\mathcal{H}_j}-\mathcal{I}-\lambda\tau\mathcal{H}_j\right)
    -\sum_{k=1}^{K}p_{J+k}\left(e^{\lambda\tau\mathcal{L}_k}-\mathcal{I}-\lambda\tau\mathcal{L}_k\right).
    \label{eq:qdrift-one-step-difference}
\end{equation}

We now bound each of the remainder terms in~\eqref{eq:qdrift-one-step-difference}.
For any Lindbladian generator $\mathcal{G}$ and $u\geq0$, Taylor's formula with integral remainder gives
\begin{equation}
    e^{u\mathcal{G}} - \mathcal{I} - u\mathcal{G}
    =\int_0^u dv\,(u-v)e^{v\mathcal{G}}\circ\mathcal{G}^2.
\end{equation}
Since $e^{v\mathcal{G}}$ is a quantum channel for $v\geq0$, its diamond norm is one. Using the submultiplicativity of the diamond norm, we therefore obtain
\begin{align}
    \left\Vert e^{u\mathcal{G}} - \mathcal{I} - u\mathcal{G} \right\Vert_{\diamond}
    &\leq \int_0^u dv\, (u-v)\left\Vert e^{v\mathcal{G}}\right\Vert_{\diamond}\left\Vert\mathcal{G}^2\right\Vert_{\diamond} \\
    &\leq \int_0^u dv\, (u-v)\left\Vert\mathcal{G}\right\Vert_{\diamond}^2 \\
    &= \frac{u^2}{2}\left\Vert\mathcal{G}\right\Vert_{\diamond}^2.
    \label{eq:exp-integral-remainder-bound}
\end{align}

Since $\Vert H_j\Vert=1$ and $\Vert L_k\Vert=1$, the submultiplicativity of the trace norm gives
\begin{equation}
    \Vert\mathcal{H}_j\Vert_{\diamond}\leq2,
    \qquad
    \Vert\mathcal{L}_k\Vert_{\diamond}\leq2.
\end{equation}
Consequently,
\begin{equation}
    \Vert\mathcal{L}\Vert_{\diamond}
    \leq \lambda\left(\sum_{j=1}^{J}p_j\Vert\mathcal{H}_j\Vert_{\diamond} + \sum_{k=1}^{K}p_{J+k}\Vert\mathcal{L}_k\Vert_{\diamond}\right)
    \leq 2\lambda.
\end{equation}

Applying~\eqref{eq:exp-integral-remainder-bound} to the three types of terms in~\eqref{eq:qdrift-one-step-difference}, we obtain
\begin{align}
    \left\Vert e^{\tau\mathcal{L}} - \mathcal{Q}_{\tau} \right\Vert_{\diamond}
    &\leq \frac{\tau^2}{2} \Vert\mathcal{L}\Vert_{\diamond}^2 +
    \frac{\lambda^2\tau^2}{2}\left(\sum_{j=1}^{J}p_j\Vert\mathcal{H}_j\Vert_{\diamond}^2+\sum_{k=1}^{K}p_{J+k}\Vert\mathcal{L}_k\Vert_{\diamond}^2\right)\\
    &\leq \frac{\tau^2}{2}(2\lambda)^2+\frac{\lambda^2\tau^2}{2}\left(4\sum_{j=1}^{J}p_j+4\sum_{k=1}^{K}p_{J+k}\right)\\
    &\leq 4\lambda^2\tau^2.
    \label{eq:qdrift-one-step-bound}
\end{align}

Finally, since both $e^{\tau\mathcal{L}}$ and $\mathcal{Q}_{\tau}$ are quantum channels, a telescoping expansion together with~\eqref{eq:qdrift-one-step-bound} gives
\begin{align}
    \frac{1}{2}\left\Vert e^{t\mathcal{L}} - \mathcal{Q}_{\tau}^{r} \right\Vert_{\diamond}
    &= \frac{1}{2}\left\Vert \left(e^{\tau\mathcal{L}}\right)^r - \mathcal{Q}_{\tau}^{r} \right\Vert_{\diamond}\\
    &\leq \frac{r}{2}\left\Vert e^{\tau\mathcal{L}} - \mathcal{Q}_{\tau} \right\Vert_{\diamond}\\
    &\leq 2r\lambda^2\tau^2\\
    &= \frac{2\lambda^2t^2}{r}.
\end{align}
Thus, if
\begin{equation}
    r \geq \frac{2\lambda^2t^2}{\varepsilon},
\end{equation}
then
\begin{equation}
    \frac{1}{2}\left\Vert e^{t\mathcal{L}} - \mathcal{Q}_{\tau}^{r} \right\Vert_{\diamond} \leq \varepsilon,
\end{equation}
which proves the proposition.

\end{document}